\documentclass[journal,10pt]{IEEEtran}
\IEEEoverridecommandlockouts
\usepackage{cite}
\usepackage{indentfirst}
\usepackage{multirow}
\usepackage{booktabs}
\usepackage{makecell}
\usepackage{float}
\usepackage{amsmath,amssymb,amsfonts}
\usepackage{algorithmic}
\usepackage{graphicx}
\usepackage{graphicx,subfigure}
\graphicspath{{./figures/}}
\usepackage{textcomp}
\usepackage{xcolor}
\usepackage{subfigure}
\usepackage{booktabs}
\usepackage{bm}
\usepackage{amsthm}
\usepackage[lined,boxed,commentsnumbered,ruled,linesnumbered]{algorithm2e}
\def\BibTeX{{\rm B\kern-.05em{\sc i\kern-.025em b}\kern-.08em
		T\kern-.1667em\lower.7ex\hbox{E}\kern-.125emX}}

\begin{document}
	\title{From Data-driven Learning to Physics-inspired Inferring: A Novel Mobile MIMO Channel Prediction Scheme Based on Neural ODE}
	\author{Zhuoran~Xiao,~\IEEEmembership{Student~Member,~IEEE, }
			Zhaoyang~Zhang,~\IEEEmembership{Senior~Member,~IEEE, }
			Zirui~Chen,~\IEEEmembership{Student~Member,~IEEE, }
			Zhaohui~Yang,~\IEEEmembership{Member,~IEEE, }\\
            Chongwen~Huang,~\IEEEmembership{Member,~IEEE, }
            and Xiaoming~Chen,~\IEEEmembership{Senior~Member,~IEEE }
		\thanks{This work was supported in part by National Key R\&D Program of China under Grant 2020YFB1807101 and 2018YFB1801104, National Natural Science Foundation of China under Grant U20A20158 and 61725104, and Provincial Key R\&D Program of Zhejiang under Grant 2023C01021. }
		\thanks{Z.~Xiao, Z.~Zhang (Corresponding Author), Z.~Chen, Z.~Yang, C.~Huang, and X.~Chen are with College of Information Science and Electronic Engineering, Zhejiang University, Hangzhou 310027, China, and with International Joint Innovation Center, Zhejiang University, Haining 314400, China, and also with Zhejiang Provincial Key Laboratory of Info. Proc., Commun. \& Netw. (IPCAN), Hangzhou 310007, China. (e-mail: \{zhuoranxiao, ning\_ming, ziruichen, yang\_zhaohui, chongwenhuang, chen\_xiaoming\}@zju.edu.cn) }
	}
	
	\maketitle
	
\begin{abstract}
		
In this paper, we propose an innovative learning-based channel prediction scheme so as to achieve higher prediction accuracy and reduce the requirements of huge amounts and strict sequential format of channel data. 
Inspired by the idea of the neural ordinary differential equation (Neural ODE), we first prove that the channel prediction problem can be modeled as an ODE problem with a known initial value by analyzing the physical process of electromagnetic wave propagation within a mobile environment. Then, we design a novel physics-inspired spatial channel gradient network (SCGnet), which represents the derivative process of channel varying as a special neural network and can obtain the gradients at any relative displacement needed for the ODE solving. With the SCGnet, the static channel at any location served by the base station is accurately inferred through consecutive propagation and integration. Finally, we design an efficient recurrent positioning algorithm based on some prior knowledge of user mobility to obtain the velocity vector and propose an approximate Doppler compensation method to make up the instantaneous angular-delay domain channel. Only discrete historical channel data is needed for the training, whereas only a few fresh channel measurements are needed for the prediction, which ensures the scheme's practicability. 
Comprehensive evaluations show that the proposed scheme is most efficient in representing, learning, and predicting mobile wireless channels.
		
\end{abstract}
	

\section{Introduction}
\subsection{Motivation}
Obtaining accurate channel state information (CSI) so as to achieve optimal system design and deployment is crucial in the future wireless communication systems \cite{9427230}. However, accurate CSI acquisition is challenging, especially in mobile communication, due to the multi-path coupling caused by the complicated scattering environment and Doppler effect \cite{9839184}. 

In the existing wireless communication system, CSI is estimated at the receiver and then feedback to the base station (BS). This scheme has at least two potential drawbacks. On the one hand, the estimated channel's accuracy is hard to guarantee if the channel parameters to be estimated are strictly limited. Meanwhile, if we increase the number of parameters to be estimated, the communication delay and signaling overhead will undoubtedly grow \cite{8395053}. 
On the other hand, there exists a nonnegligible transmission delay for CSI data feedback. Therefore, the instantaneous channel had already changed when the BS received the feedback CSI data, which will inevitably cause a loss of accuracy.
Deterministic channel synthesis, such as the ray-tracing method, is another approach for obtaining CSI. This approach is based on the specific electromagnetic (EM) wave propagation within a natural environment. Owing to the analytical framework rooted in electromagnetic theory, this method considers the field properties of the antennas and the propagation paths. Although it can achieve favorable accuracy with excellent cost-and-time efficiency, it is usually hard to be applied in practice since it needs to know the complete knowledge of the natural environment.

In a specific communication scenario, a base station always serves a certain fixed area in which the number and locations of the scatterers are also fixed. The wireless channel is then mainly determined by the scatterers distributed therein and the corresponding multi-path propagation behavior of EM waves, including specular reflection, diffusion, diffraction, and blocking. Although there are inevitable variations in channel characteristics due to the small-scale perturbation of transmission media, the long-term multi-path components of a channel can still be regarded as quasi-static given that the scatterers within that area do not change \cite{9625179,4536852,9791407,4735375}. Moreover, when considering the mobile scene, the wireless channel can be calculated by compensating the Doppler shift to the static channel. So, it is reasonable that with the velocity vector and a set of randomly measured static CSI samples, we can predict the CSI of a user equipment (UE) moving at a certain point within the served area. Moreover, there are always a large number of channel measurements with corresponding locations produced by the communication processes occurring in the served area of the BS. These data can be potentially made use of, and thus the signaling overhead and measurement cost can be significantly reduced. This inspires us to study the prediction of the mobile channel within a specific communication environment. 

Due to the complexity characteristic of multi-path channels, we resort to learning-based methods for solving this prediction problem for the higher representing ability due to larger parameter space. Basically, our goal is to design a learning system that can automatically learn the fixed scattering environment and implicitly represent the physics process of the wireless channel changing within space. The whole learning system can be trained at the training stage with only a few discrete static channel samples randomly recorded in the historical communication process. Our goal is to significantly reduce the difficulty of data acquisition and the demand for data volume. At the inferring stage, only a sequence of channels obtained in a past period of time is needed. All the other key factors, such as the UE's motion vector and the UE's position, are automatically learned. The mobile channel frequency response (CFR) is the output target of the whole system.
	
\subsection{Related Works}

In the past few years, comprehensive works have been done in CSI prediction. In \cite{8932272,8116491}, the noise from the measured channel impulse response (CIR) is removed by neural networks, and the principal component analysis (PCA) was utilized to exploit the features and structures of the channel. In \cite{2016OnThe}, based on received signal strength measurements and a 3-D map of the propagation environment, the support vector machine (SVM) was used to predict the path loss for smart metering applications. Taking the receiver (RX) and transmitter (TX) heights, TX-RX distance, carrier frequency, and intercede range as the inputs, the path loss is predicted using the radial basis function neural network (RBF-NN) in \cite{554747}. In \cite{6748900}, the fully connected neural network (FNN) and RBF-NN were combined with ray launching in complex indoor environments for predicting the intermediate points in the ray launching algorithm and further decreasing the computation complexity. The authors in \cite{7590098} take TX-RX distance and diffraction loss as the inputs, and the multi-layer perceptron (MLP) was applied to predict the received signal strength. 

Most existing works on channel prediction incline to treat it as a sequence prediction problem. Two conventional prediction approaches, namely parametric model \cite{6945858} and auto-regressive model \cite{841729} have been proposed through statistical modeling of a wireless channel as a set of radio propagation parameters. Since there is a nonnegligible gap between the model and the real wireless channel, the potential achievable prediction accuracy is generally unfavorable, which highly limits their actual application. Motivated by the success of machine learning algorithms in time-series prediction, some recent works have attempted to apply learning-based methods for channel prediction. In \cite{2002Recurrent,6755477}, the recurrent neural network (RNN) is applied to predict narrow-band single-antenna channels. In \cite{9128426}, RNN is further replaced by the gated recurrent unit (GRU) and a long short-term memory (LSTM). Through viewing the sequential CSI as an image, CNN+LSTM structure is used for channel prediction in \cite{9277535}. In \cite{9569281}, a generative periodic activator-enabled Auto-Encoder LSTM network is proposed as an improved method of normal LSTM. The learning structure CNN+AR is proposed as a substitution of LSTM in \cite{8979256}. In \cite{8904286}, the Seq2Seq network is applied. Besides, the ODE-RNN structure is introduced in \cite{9978073}. It is worth mentioning that the basic idea behind \cite{9978073} is quite different from this paper since therein the ODE-RNN is proposed as an upgrade of normal RNN for solving sequence prediction problems while here we attempt to predict channel with only discrete CSI samples employing an entirely newly designed network structure.

Although the existing learning-based methods for channel prediction can achieve better prediction accuracy than traditional methods, they are still insufficient for practical application. First of all, all the methods, including \cite{9978073}, require a dataset composed of a series of channel sequences. However, it should be emphasized that storing large quantities of sequence data is not compatible with the existing wireless communication protocols and will cause huge challenges to storage capabilities for the BS and potential data transmission overhead if there exists a cloud server for model training. All the above facts make the state-of-the-art infeasible in practice. In fact, in real scenarios, only the discrete CSI recorded at a certain position rather than a data sequence is available since this kind of data can be easily obtained from historical records of those service vendors, such as the communication operators. The proposed methods in this paper are designed to deal with the setting that the BS only stores some discrete CSI data sampled randomly in the served area. Thus, the difficulty in obtaining and storing training data can be greatly reduced and can be much more practical and flexible at the deployment stage. Secondly, the accuracy of sequence prediction by a purely data-driven learning network highly relies on the numerical correlation and smoothness of the data sequence. Therefore, when the interval of data sequence increases, the prediction accuracy of those methods significantly degrades. The main reason is that the network structures used therein mainly focus on the common data correlation between the discrete elements of all sequences while overlooking the fact that the actual channel changing is a unique physical continuous process with specific properties. Therefore, no additional prior information is integrated into the learning network, which limits the prediction accuracy.
	
\subsection{Contributions}
In fact, the change of CSI comes from the position change, which causes the change of propagation paths. The changing of CSI is essentially a physics process driven by the propagation characteristics of electromagnetic waves in space and users' movement characteristics, which contains a lot of implicit prior knowledge. Thus, rather than predict the channel in the original high dimension and unsmooth data space by learning the correlation between discrete sampling points, we propose to learn the implicit ordinary differential equations representing the complete physics process of continuous channel changing. Thus, we come up with a novel channel prediction scheme that is composed of three segments. Firstly, we propose a specially designed physics-inspired SCGnet to represent the gradients of the derivative process of channel varying with respect to space. Further, the output of which is then integrated to obtain the targeted output in the way of Neural ODE. Discrete channel pairs with their position coordinates are used to train the network through the adjoint method \cite{2018Optimization}. Then, we propose an iterative algorithm for user positioning and motion vector extraction based on a novel-designed positioning neural network. Based on this, we transfer the original prediction problem from a high and unsmooth data space to a low and smooth one, significantly reducing learning difficulty. Thirdly, an approximating Doppler compensation method making use of the characteristic of the angular-delay domain channel is proposed. Then, the targeted mobile channel can be obtained by compensating the Doppler phase shift on the predicted static channel.    
		
To summarize, our main contributions are as follows:
\begin{itemize}
\item 
We prove that the gradient of CSI with respect to space is only decided by the CSI itself and spatial direction. So, we propose to model channel prediction as an ODE problem with a known initial value. The Neural ODE learning structure is proposed to solve the problem, and a physics-inspired SCGnet is designed to implicitly represent the ODE function. Only discrete samples are needed to train the network, which can be easily obtained from the historical communication data. 

\item 
An iterative algorithm for user positioning and motion vector extraction is proposed based on a specially designed positioning neural network. By making use of the prior knowledge of the user's mobility and the ODE-guided change of its underlying scattering environment, the positioning accuracy can far exceed that of the positioning methods with a single discrete point. With the predicted user's velocity vector, an approximate Doppler compensation algorithm is proposed for calculating the mobile channel. 

\item Key performance measures, including network scale, prediction accuracy, and dependence on training data, are comprehensively evaluated and compared with the state-of-the-art results, which shows that our proposed scheme is much more efficient in learning, representing, and predicting wireless channels in a given communication environment. 
\end{itemize}
	
The remainder of this paper is organized as follows. The channel model and the problem formulation of channel prediction are introduced in section \ref{chap:Problem}. The reason why we introduce Neural ODE for channel prediction is given in section \ref{chap:Neural ODE}. Section \ref{chap:scheme} introduces the whole channel prediction scheme proposed by us, which is composed of three parts. Numerical evaluation from different perspectives is provided in section \ref{chap:results}. Section \ref{chap:conclusion} draws the conclusion.
	
\section{Problem Formulation}\label{chap:Problem}
\subsection{Channel Model} \label{chap:channel}
Without losing generality, we assume that the base station (BS) is equipped with a uniform linear array (ULA) with half-wavelength spacing between two adjacent antennas. Besides that, orthogonal frequency-division multiplexing (OFDM) modulation is adopted. The UE has a single omnidirectional antenna. The BS has ${N_t}$ antennas, and there are ${N_c}$ subcarriers for OFDM signals. 
The static channel frequency response (CFR) for each subcarrier can be formulated as
\begin{equation}\label{Sye_CFR_eq1}
{\bf{h}}[l] = \sum\limits_{p = 1}^{N_{\rm{path}}} {{\alpha _p}{\bf{e}}({\theta _p})} {e^{ - j2\pi {{{d_p}} \over {{\lambda _l}}}}},
\end{equation}
and if the UE's mobility is considered, the mobile CFR for each subcarrier can be formulated as 
\begin{equation}\label{Sye_CFR_eq2}
{{\bf{h}}_{{\rm{mobile}}}}[l] = \sum\limits_{p = 1}^{N_{\rm{path}}} {{\alpha _p}{\bf{e}}({\theta _p})} {e^{ - j2\pi [{{{d_p} + {v_u}\cos ({\theta _v} - {\theta _p}){\tau _p}} \over {{\lambda _l}}}]}},
\end{equation}
where $l$ denotes the subcarrier index, ${N_{\rm{path}}}$ is the total number of propagation paths, ${{\alpha _p}}$ is the path loss of $p$th path, ${{\theta_p}}$ is the angle of arrival, ${{d_p}}$ is the length of propagation path, ${{v_u}}$ denotes the velocity of user, ${{\theta _v}}$ is the direction angle of the velocity vector, ${{\tau _p}}$ is the propagation delay and ${{\lambda _l}}$ is the wavelength of corresponding subcarrier. In equation \eqref{Sye_CFR_eq1} and \eqref{Sye_CFR_eq2}, ${\bf e}(\theta)$ denotes the array response vector of the ULA given by
\begin{equation}
	{\bf{e}}(\theta ) = {[1,{e^{ - j2\pi {{d\cos (\theta )} \over \lambda }}},...,{e^{ - j2\pi {{({N_t} - 1)d\cos (\theta )} \over \lambda }}}]^T},
\end{equation}
where $d$ is the interval between two adjacent antennas and $\lambda $ is the wavelength. Thus, the overall CFR matrix of the channel between the BS and the user can be expressed as 
\begin{equation}
	{\bf{H}} = [{\bf{h}}[1],{\bf{h}}[2],...,{\bf{h}}[{N_c}]].
\end{equation}

\subsection{Angular-Delay Domain Channel Representation} \label{chap:Angular}
Angular-Delay domain channel matrix is a transformation of the CSI computed by multiplying it with two discrete Fourier transform (DFT) matrices \cite{2019Fingerprint}. Defining the DFT matrix $\bm{V} \in {{\Bbb C}^{{N_t} \times {N_t}}}$ as 
\begin{equation}
	{[\bm{V}]_{z,q}} \triangleq \frac{1}{{\sqrt {{N_t}} }}{e^{ - j2\pi \frac{{(z(q - \frac{{{N_t}}}{2}))}}{{{N_t}}}}},
\end{equation}
and $\bm{F} \in {{\Bbb C}^{{N_c} \times {N_c}}}$ as
\begin{equation}
	{[\bm{F}]_{z,q}} \triangleq \frac{1}{{\sqrt {{N_c}} }}{e^{ - j2\pi \frac{{zq}}{{{N_c}}}}}.
\end{equation}
Then, the Angular-Delay domain channel matrix $G$ is defined as 
\begin{equation}
	\bm{G} = {\bm{V}^H}\bm{H}\bm{F}
\end{equation} 
The matrix $\bm{G}$ can be considered as a uniform grid of ${N_t}$ AoAs over $[0,\pi ]$ and ${N_c}$ delay points over $[0,{{{N_c}} \over B}]$, where $B$ denotes the carrier bandwidth. Thus, In this matrix, the $(z,q)$ element represents the response of the propagation path with ${z^{th}}$ AOA and ${q^{th}}$ delay, which can be written as
\begin{equation}\label{eq1}
{\widetilde \theta _z} = \arccos \left[ {{2 \over {{N_t}}}(z - \left\lfloor {{{{N_t} - 1} \over 2}} \right\rfloor )} \right],z = 0, \ldots ,{N_t} - 1,
\end{equation}
\begin{equation}\label{eq2}
{\widetilde \tau _q} = {q \over B},q = 0, \ldots ,{N_c} - 1.
\end{equation}
In practice, the true AoAs and delay points usually do not always lie exactly on the grid points. In this case, there will be mismatches between the true AoAs, delay points, and the nearest grid point. Thus, the angular-delay channel response can be regarded as being composed of ${N_t} \times {N_c}$ equivalent paths. Due to the sparsity of the matrix $\bm{G}$, most equivalent paths' responses are close to zero. For each equivalent propagation path with $\theta $ AOA and $\tau$ delay, the complex path response is coupled by multiple subcarrier coupling, which can be written as 
\begin{equation}
{g_{z ,q }} = {\alpha _{\theta, \tau}}\sum\limits_{k = 1}^{{N_c}} {{e^{ - j2\pi {f_k}\tau }}} ,
\end{equation}

\subsection{Problem Formulation of Channel Prediction} \label{chap:channel prediction}
The goal of channel prediction is to forecast the CSI at the current and future time slots by taking full use of CSI sequence information obtained in the previous time period \cite{2020Recurrent}. Assume that the BS stores the CSI estimated in the past $n$ time slots, which can be denoted by $\{ {\bf{H}}[1],{\bf{H}}[2],...,{\bf{H}}[n]\} $. The CSI in the next time slot should be predicted, denoted by $\mathop {\bf{H}}\limits^ \wedge  [n + 1]$. 
Thus, the CSI prediction problem can be presented as 
\begin{equation}
	\{ {\bf{H}}[1],{\bf{H}}[2],...,{\bf{H}}[n]\}  \to \mathop {\bf{H}}\limits^ \wedge  [n + 1].
\end{equation}

\section{The ODE Representation of Spatial Gradient of Static Channel}\label{chap:Neural ODE}

\subsection{Overview of Neural ODE}
Neural ODE is a specially designed learning structure that has been proven to have a strong ability to represent dynamic progress, which can be originally described mathematically by ordinary differential equations \cite{NEURIPS2018_69386f6b}. In a typical Neural ODE, the state of the hidden layer is defined by the solution of the following equation,
\begin{equation}
	{{dy(x)} \over {dx}} = f(y(x),x,{\bf{I}}{(x),}{\bf{\boldsymbol\theta}}),
\end{equation}
where ${y}(x) \in {\mathbb{R}^D}$ denotes the hidden layer, $D$ is the dimension of variables in the hidden layer, ${\bf{I}}(x)$ is the input, $x$ denotes the independent variable, $f( \cdot )$ denotes a neural network with parameter $\boldsymbol\theta$. Therefore, after the network completes training, the forward calculation becomes an ordinary differential equation problem with known initial values. Thus, the forward calculation can be solved with constant or adaptive steps by numerical methods called ODE Solver. The process of ODE Solver can be written as 
\begin{equation}
	y({x_0}) = {y_0},
\end{equation}
\begin{equation}
	{{y_{\rm{output}}}} = \text{ODE Solver}(f_{\boldsymbol\theta},{y_0},{x_0},{x_{end}}),
\end{equation}
Where ${y_{\rm{output}}}$ is the output of the Neural ODE, ${y_0}$ is the initial value, ${x_0}$ is the integral starting point, and ${x_{end}}$ is the integral endpoint.

The Euler method is one of the most commonly used ODE Solvers among all the numerical methods. The procedure of the Euler method can be presented by 
\begin{equation}
	{y_{x + \Delta x}} = {y_x} + \Delta x \times f({y_x},{\boldsymbol\theta}),
\end{equation}
where $\Delta x$ is the step length used to adjust the accuracy and computation cost. In addition to the Euler method, there are some high-order solvers with adaptive step sizes. Choosing a proper ODE Solver is determined by the computational cost and accuracy trade-offs. 

In order to ensure high calculation accuracy, the ODE Solver's step size is set as a  small value, which indicates that directly using the gradient back propagation algorithm to calculate the gradient of parameters will introduce an enormous calculation cost. To solve this problem, the adjoint method is proposed in \cite{2018Optimization} to transform the gradient calculation into another ODE problem, which can be solved by the ODE solver with a low computational cost. The adaptive checkpoint adjoint method is further proposed in \cite{2020Adaptive} by adding a small amount of storage cost to handle the calculation error introduced in the adjoint method.

\subsection{The Limitation of Conventional Network for Channel Prediction}
As shown in \eqref{Sye_CFR_eq1}, the CSI matrix is determined by multipath response components interwinding. For a wireless channel, the variation mainly comes from two parts, i.e., the change of electromagnetic wave propagation paths caused by the changing spatial position and frequency selective fading caused by the Doppler effect. Due to the apparent magnitude difference of electromagnetic wavelength scale relative to the spatial changing scale of mobile users, the phase changing of channel response is fast across space. Thus, the path response shows prominent non-smooth characteristics on the spatial coordinate axis. Considering the interwinding of multiple path responses and the high-dimensional characteristics of channel matrix caused by multi-carriers and multi-antennas, from the perspective of data characteristics, the channel prediction problem is essentially a prediction problem dealing with high-dimensional and extremely non-smooth data.

For those general sequential networks such as RNN and LSTM, the dimension of data space and the data smoothness directly decide the performance of the prediction. In existing works, solving such a complex problem requires not only a high level of network learning ability but also a large amount of training data. Apparently, those requirements are hard to meet in a practical system. However, as shown in Fig. \ref{propagation}, from the perspective of the physics process of electromagnetic wave propagation, the changing of the channel is a pure physics driving process. Besides that, the propagation of the electromagnetic wave in space has some fixed physical characteristics. Thus, some unique prior information can be used by designing targeted learning architecture to improve prediction accuracy, save training data, and improve system flexibility.

\begin{figure}
	\centering
	\includegraphics[width=0.4\textwidth]{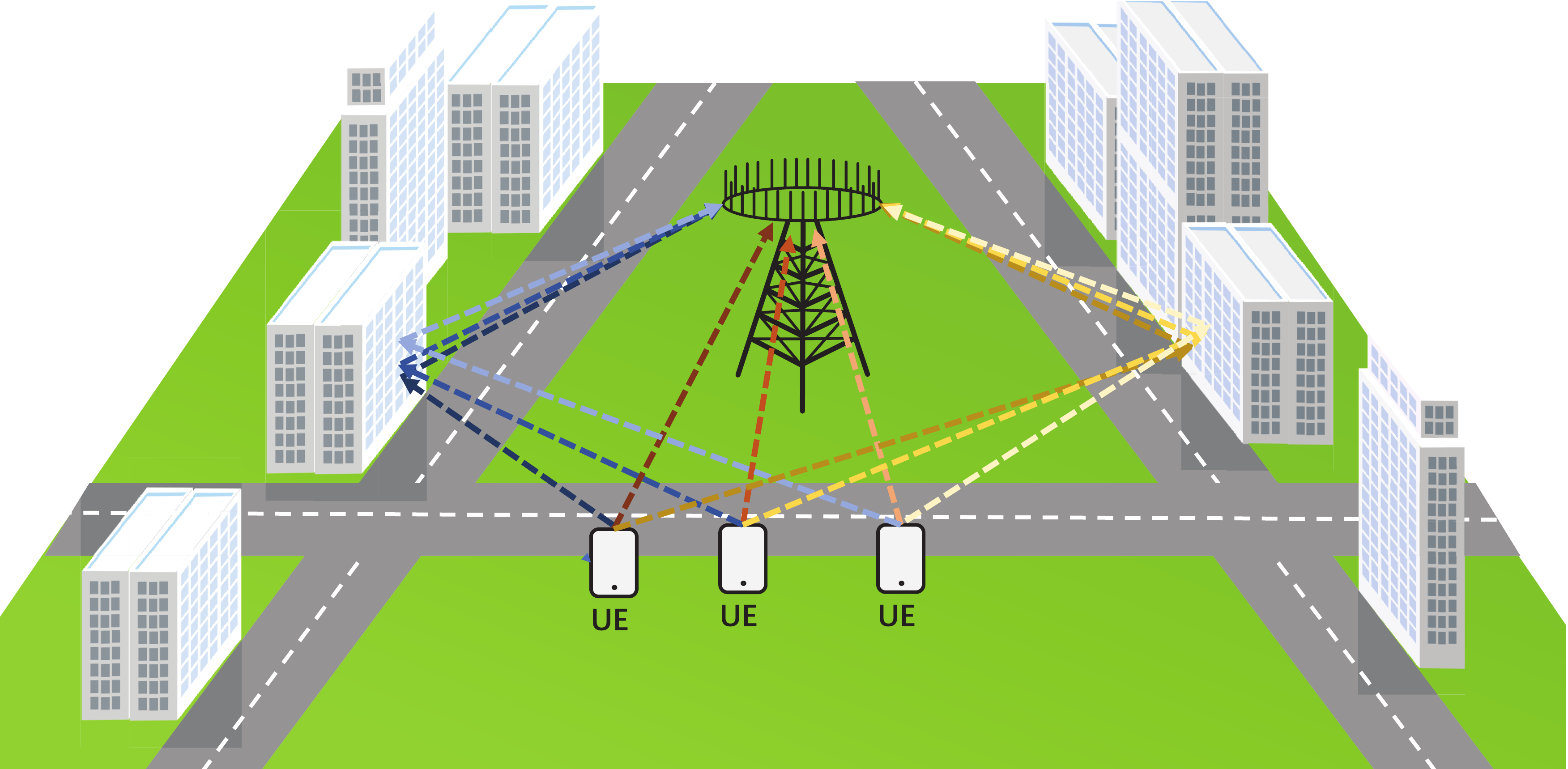}
	\vspace{-.5em}
	\caption{When UE's position change continuously in space, the changing process of propagation paths are also continuous and highly smooth.}
	\vspace{-.5em}
	\label{propagation}

\end{figure}

\subsection{Represent The Process of Channel Change with ODE}

In order to better explain the advantages of the proposed learning architecture, we provide the following proposition to show the theoretical basis of the network design proposed for CSI prediction.

\newtheorem{lemma}{\textbf{Proposition}}
\begin{lemma}\label{lemma}
For a quasi-static scattering environment, if the bidirectional mapping between of CSI $ {{\bf{H}}_u}$ and UE's position $ {{\bf{x}}_u}$ exists, the derivative of the mobile channel with respect to space is only related to the CSI at a certain point and UE's motion direction. 
\end{lemma}

\begin{proof}[\textbf{Proof}]
According to \eqref{Sye_CFR_eq1}, the real and imaginary parts of channel response of the $u$th subcarrier and $v$th antenna can be respectively written as
	\begin{equation}
		\begin{aligned}
			h_{u,v}^{\rm{real}} = \sum\limits_{p = 1}^{{N_{\rm{path}}}} {{\alpha _p}} \cos ( - 2\pi {{{d_p}} \over {{\lambda _u}}} - \pi (v - 1)\cos {\theta _p}),
		\end{aligned}
	\end{equation}
	and 
	\begin{equation}
		\begin{aligned}
		h_{u,v}^{\rm{imag}} = \sum\limits_{p = 1}^{{N_{\rm{path}}}} {{\alpha _p}} \sin ( - 2\pi {{{d_p}} \over {{\lambda _u}}} - \pi (v - 1)\cos {\theta _p}),
		\end{aligned}
	\end{equation}
	where ${N_{\rm{path}}}$ is the total amounts of propagation paths, ${d_p}$ is the propagation length of path $p$, ${\alpha _p}$ is the path loss and ${\lambda _u}$ is the wavelength of the $u$th subcarrier. In particular, ${\alpha _p}$ is an inverse proportional function of ${d_p}$. Thus, we denote ${\alpha _p} = {{{\xi _p}} \over {{d_p}}}$, where ${\xi _p}$ is only related to the characteristics of electromagnetic materials. Also, we denote $-2\pi {{{d_p}} \over {{\lambda _u}}} = \rho {d_p}$ and $ - \pi (v - 1)\cos {\theta _p} = b$. Taking the real part of channel response for illustration, its derivative with respect to spatial displacement is
	\begin{equation}
		\begin{aligned}
			&{{\partial h_{u,v}^{\rm{real}}} \over {\partial \bf{m}}} = {{\partial h_{u,v}^{\rm{real}}} \over {\partial {d_p}}}{{\partial {d_p}} \over {\partial \bf{m}}}\\& = [\sum\limits_{p = 1}^{N_p} {{{{\xi _p}} \over {{d_p}}}} \sin ( \rho {d_p} + b)\rho - \sum\limits_{p = 1}^{N_p} {{{{\xi _p}} \over {{d_p}}}} \cos (\rho {d_p} + b){1 \over {{d_p}}}]{{\partial {d_p}} \over {\partial \bf{m}}} \\& = \sum\limits_{p = 1}^{N_p} {{{\xi _p}} \over {{d_p}}} (\sin ( \rho {d_p} + b)\rho - \cos (\rho {d_p} + b){1 \over {{d_p}}}){{\partial {d_p}} \over {\partial \bf{m}}},
		\end{aligned}
	\end{equation}
	\begin{equation}
	{{\partial {d_p}} \over {\partial {\bf{m}}}} = \cos ({\theta _m} - {\theta _p}),
	\end{equation}
	where ${\theta _m}$ denotes the motion direction.
	
	Similarly, the derivative of the imaginary part can be given by 
	\begin{equation}
		{{\partial h_{u,v}^{\rm{imag}}} \over {\partial \bf{m}}} =  - \sum\limits_{p = 1}^{N_p} {{{\xi _p}} \over {{d_p}}} (\cos ( \rho {d_p} + b)\rho + \sin (\rho {d_p} + b){1 \over {{d_p}}}){{\partial {d_p}} \over {\partial \bf{m}}}.
	\end{equation}
	It is obvious that all the variables in the above equations, such as ${d_p}$, ${\xi _p}$ and ${\rho}$ are decided by the UE's position ${{\bf{x}}_{\bf{u}}}$ with the quasi-static scattering environment. Besides, the position is determined by the UE's CSI ${{\bf{H}}_u}$, based on the bidirectional mapping assumption of $
		\{ {{\bf{H}}_u}\}  \to \{ {{\bf{x}}_{\bf{u}}}\}$
		. Thus, the whole formula only has two variable $ {{\bf{H}}_u}$ and ${\theta _m}$. Therefore, this proposition is proved.
\end{proof}
It should be noted that, in practical wireless communication environments, the bidirectional mapping between ${{\bf{H}}_u}$ and ${\bf{x}}_{\bf{u}}$ holds with high probability due to the rich scattering and the large-scale antenna arrays deployed therein \cite{8292280}. Consequently, with the above proposition, it is reasonable that we use an ordinary differential equation to represent the changing process of the wireless channel with respect to space
\begin{equation}
{{\partial {\bf{H}}} \over {\partial {\bf{m}}}} = f({\bf{H}},{\theta _m}).
\end{equation}
Thus, the under-predicted channel can be calculated by forward integration, and any known static channel at a certain position obtained in a historical communication process can be used as the initial value 
\begin{equation}
{{\bf{H}}_{\rm{predict}}} = {{\bf{H}}_0} + \int_0^s {f({\bf{H}},{\theta _m})} d{\bf{m}},
\end{equation}
where ${{\bf{H}}_{\rm{predict}}}$ is the static channel to be predicted, ${{\bf{H}}_0}$ is the known channel and $s$ is the distance between the two points. Since the function $f( \cdot )$ is not explicitly known, a neural network is needed to represent it implicitly. Further, the complete integration process can be represented by a neural ODE.

\section{The Physics-Inspired Mobile Channel Prediction Scheme} \label{chap:scheme}

\subsection{The Overall Process of Mobile Channel Prediction based on Neural ODE}
As shown in Fig. \ref{process}, the whole mobile channel prediction system consists of three parts. We first obtain the UE's position and velocity vector through the iteration algorithm, which will be introduced in the following text. Then, we feed the position coordination to the Neural ODE structure trained with the historical static channel. Finally, the mobile channel is obtained through Doppler compensation with a known velocity vector. In principle, the Neural ODE module is used to learn the implicit information related to the scattering environment and further generate CSI. The other parts are used to learn the implicit information related to the mobile UE. Compared with those purely data-driven schemes, we design networks separately for different functions required in the whole process, which reduces the difficulty of network learning and the dependence on data volume.

\begin{figure}[htb!]
	\centering
	\includegraphics[width=0.4\textwidth]{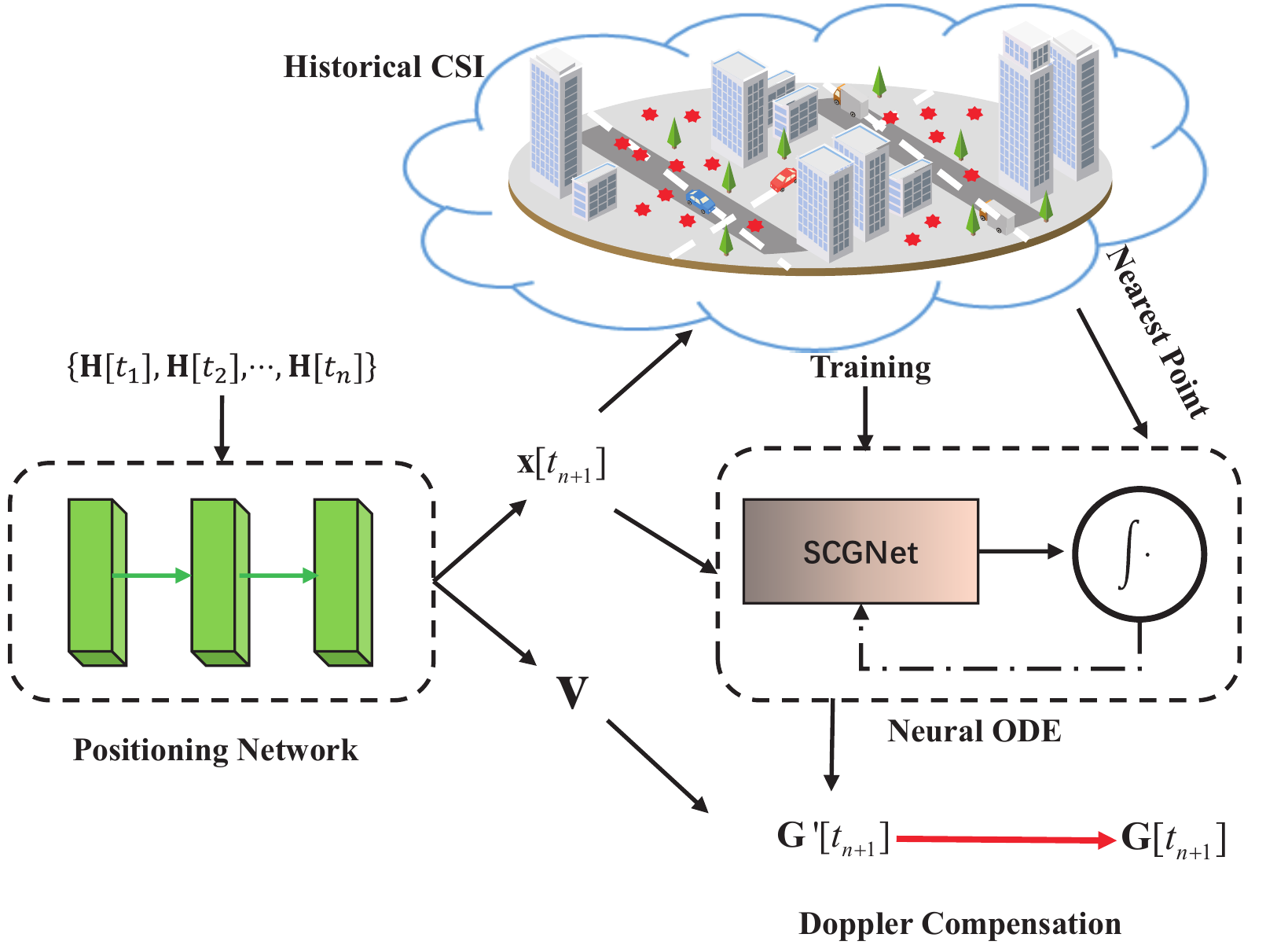}
	\caption{The whole process of mobile channel prediction based on Neural ODE.}
	\label{process}
\end{figure}
\subsection{The Learning Structure of Spatial Channel Gradient Network}
For each matrix element with $\theta $ AOA and $\tau$ delay in the angular-delay domain channel, the real and imaginary parts can be written as
\begin{equation}
	g_{\tau ,\theta }^{\rm{real}} = {\alpha _{\tau ,\theta }}\sum\limits_{k = 1}^{{N_c}}  \cos ( - 2\pi {{{d_p }} \over {{\lambda _k}}}).
\end{equation}
\begin{equation}
	g_{\tau ,\theta }^{\rm{imag}} = {\alpha _{\tau ,\theta }}\sum\limits_{k = 1}^{{N_c}}  \sin ( - 2\pi {{{d_p }} \over {{\lambda _k}}}).
\end{equation}
Here, $d_p$ is indeed decided by $\tau$ and $\theta$. As mentioned above, ${\alpha _{\tau ,\theta }}$ is an inverse proportional function of ${d_p}$. We denote ${\alpha _{\tau ,\theta }} = {{{\xi _p}} \over {{d_p}}}$, where ${\xi _p}$ is only related to the characteristics of electromagnetic materials. Also, we denote $-2\pi {{{{d_p}} \over {{\lambda _k}}}} = \rho_k {d_p}$. Since in a practical system, wavelengths of different subcarriers can be considered approximately equal. So, we have ${\rho _k}{d_p} \approx \rho {d_p}$. Then, the gradient of the channel relative to space can be written as
\begin{equation}\label{real}
\begin{aligned}
		&{{\partial g_{\tau,\theta}^{\rm{real}}} \over {\partial \bf{m}}} = {{\partial g_{\tau,\theta}^{\rm{real}}} \over {\partial {d_p}}}{{\partial {d_p}} \over {\partial \bf{m}}}\\& = [\sum\limits_{k = 1}^{N_c} {{{{\xi _p}} \over {{d_p}}}} \sin ( \rho {d_p} )\rho - \sum\limits_{k = 1}^{N_c} {{{{\xi _p}} \over {{d_p}}}} \cos (\rho {d_p} ){1 \over {{d_p}}}]{{\partial {d_p}} \over {\partial \bf{m}}} \\& = ( - {1 \over {{d_p}}}g_{\tau ,\theta }^{\rm{real}} + \rho g_{\tau ,\theta }^{\rm{imag}}){{\partial {d_p}} \over {\partial {\bf{m}}}},
\end{aligned}
\end{equation}
\begin{equation}\label{imag}
	\begin{aligned}
		&{{\partial g_{\tau,\theta}^{\rm{imag}}} \over {\partial \bf{m}}} = {{\partial g_{\tau,\theta}^{\rm{imag}}} \over {\partial {d_p}}}{{\partial {d_p}} \over {\partial \bf{m}}}\\& = [-\sum\limits_{k = 1}^{N_c} {{{{\xi _p}} \over {{d_p}}}} \cos ( \rho {d_p} )\rho - \sum\limits_{k = 1}^{N_c} {{{{\xi _p}} \over {{d_p}}}} \sin (\rho {d_p} ){1 \over {{d_p}}}]{{\partial {d_p}} \over {\partial \bf{m}}} \\& = ( - {1 \over {{d_p}}}g_{\tau ,\theta }^{\rm{real}} - \rho g_{\tau ,\theta }^{\rm{imag}}){{\partial {d_p}} \over {\partial {\bf{m}}}}.
	\end{aligned}
\end{equation}

We take \eqref{real} as an example to demonstrate our ideas behind the network design. After we transfer the CSI in the angular-delay domain and calculate its gradient of space in such a way, it can be found that such a mathematical expression has some clear natural advantages for us to design the representing networks. First of all, in existing works, the original complex channel matrix must be split into real and imaginary parts. Then, we need to concatenate them together due to the constraint that the neural network can only handle real numbers. However, such an operation significantly destroys the internal correlation of input data. On the contrary, in \eqref{real} and \eqref{imag}, the expression on the right side is directly composed of the weighted sum of the real and the imaginary parts, respectively. Therefore, dividing the original complex number into real and imaginary parts does not destroy the internal correlation of data but also facilitates network learning. Secondly, as mentioned above, the original channel is high-dimensional complex data. However, in the equation, only ${1 \over {{d_p}}}$, which denotes the reciprocal of propagation path lengths, needs to be learned. Unlike existing works whose networks need to recover high-dimensional complex data (channel matrix) as the output, our network only needs to learn a simple mapping from high-dimensional complex data to low-dimensional simple data. Thus, the requirements for network parameters scale and complexity cost are significantly reduced. It is worth mentioning that the changes in the delay and angle of propagation of paths are negligible for a propagation path since the sampling interval is usually not long enough. Therefore, in the process of channel changing between two points, the indices of the matrix element representing a certain propagation path are fixed, ensuring our methods' robustness.

\begin{figure*}
	\centering
	\includegraphics[width=0.55\textwidth]{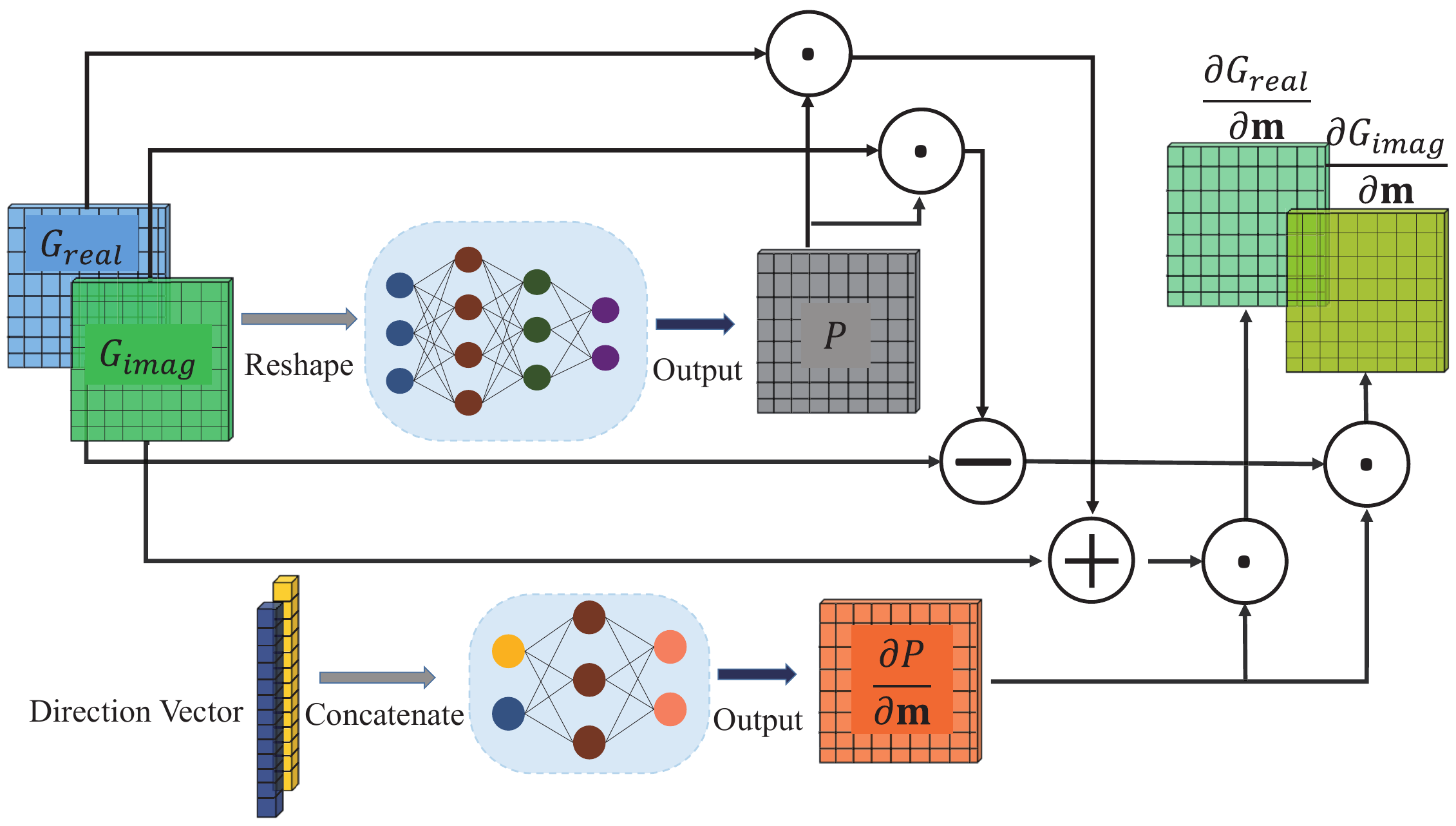}
	\vspace{-.5em}
	\caption{The network structure of physics-inspired SCGnet.}
	\vspace{-.5em}
	\label{ODE}
\end{figure*}

The overall network design is shown in Fig. \ref{ODE}. The whole structure is composed of two networks. The first network, called the scattering learning network, adopts the fully connected structure that takes the reshaped CSI matrix $\bm{G} \in {N_a} \times {N_c} \times 2$ as input, where ${N_a}$ denotes the number of antennas and ${N_c}$ denotes the number of subcarriers, and intends to obtain the matrix $P \in {N_a} \times {N_c}$ composed of propagation length $ - {1 \over {{d_p}}}$. The other network, called the direction embedding network, takes the direction vector ${\theta _{direction}} \sim [ - \pi ,\pi ]$ as input to obtain the matrix ${{\partial P} \over {\partial {\bf{m}}}} \in {N_a} \times {N_c}$ composed of ${{\partial {d_p}} \over {\partial {\bf{m}}}}$ as we have proved that ${{\partial {d_p}} \over {\partial {\bf{m}}}}$ is only decided by the integration direction. Rather than using the angle value as input, we use the sine and cosine values. It can also be viewed as a special case of position encoding, which is commonly used in state-of-the-art such as the transformer. Moreover, since the output of this network, which represents the gradient of channel relative to spatial displacement, is a vector of dimension ${N_a} \times {N_c} \times 2$, we extend the input dimension of ${N_a} \times 2$ by simple duplicating to balance the dimension of input and output for better learning efficiency. Then, the final output can be directly obtained through simple addition and subtraction with two times Hadamard production. Since the forward propagation strictly obeys the calculation process of expression (\ref{real}) - (\ref{imag}), we achieve fusing the channel prior to the network design.

\begin{algorithm}
	\small 
	\DontPrintSemicolon
	\SetAlgoLined
	\KwIn {Position coordinates ${{\bf{x}}_u}$ of UE to be predicted.\\
	}
	\KwOut {predicted static channel ${\bf{G}}({{\bf{x}}_u})$}
	Train the SCGnet with the historical channel dataset.\;
	Search for the nearest point in the database. Get the Position coordinates ${\bf{x}}'$ and channel in the angular-delay domain ${\bf{G}}({\bf{x}}')$ of the nearest point in the database. Get the interval $s$.\;
	Obtain the integration direction ${\theta _m}$ with ${\bf{x}}'$ and ${{\bf{x}}_u}$.\;
	$s' \gets 0$.\;
	\While{$s' \le s$}{
		Input ${\bf{G}}({\bf{x}}')$ and ${\theta _m}$ to the SCGnet and get the output ${{\partial {\bf{G}}({\bf{x}}')} \over {\partial {\bf{m}}}}$.\;
		${\bf{G}}({\bf{x}}') \gets {\bf{G}}({\bf{x}}') + {{\partial {\bf{G}}({\bf{x}}')} \over {\partial {\bf{m}}}} \times \Delta s$.\;
		$s' \gets s'+\Delta s$.\;
	}
	\caption{\textbf{Static Channel Prediction with Neural ODE (Euler method)}}
	\label{algorithm1}
\end{algorithm}
The overall training and inferring process of the proposed SCGnet are based on the Neural ODE scheme. In the training stage, the method for building a training dataset contains the following steps. Firstly, for each recorded CSI sample, find its $z$ nearest neighbor points and build $z$ data pairs. Then, calculate the vector's length and angle with the selected sample's position as the starting point and its neighbor as the ending point. Thus, the CSI of the selected CSI sample and the vector's angle are set as the input of the SCGnet. The length of the vector is set as the integral length, and the CSI of its neighbor points as the targeted output of the Neural ODE. It is worth mentioning that $z$ is selected according to the sampling density of historical data. More specifically, the value of $z$ needs to be smaller when the sampling density is relatively low since a larger integral will lead to greater error and influence the convergence of learning. As shown in Fig. \ref{integration}, in the inferring stage, we find the point closest to the position of the channel to be predicted in the historical database and then set its value as the initial value. Then, the channel to be predicted can be obtained by forward integration. Supposing that the ODE Solver adopts the Euler method, the whole process of static channel prediction is summarized in Algorithm \ref{algorithm1}. 

\begin{figure}[b!]
	\centering
	\includegraphics[width=0.4\textwidth]{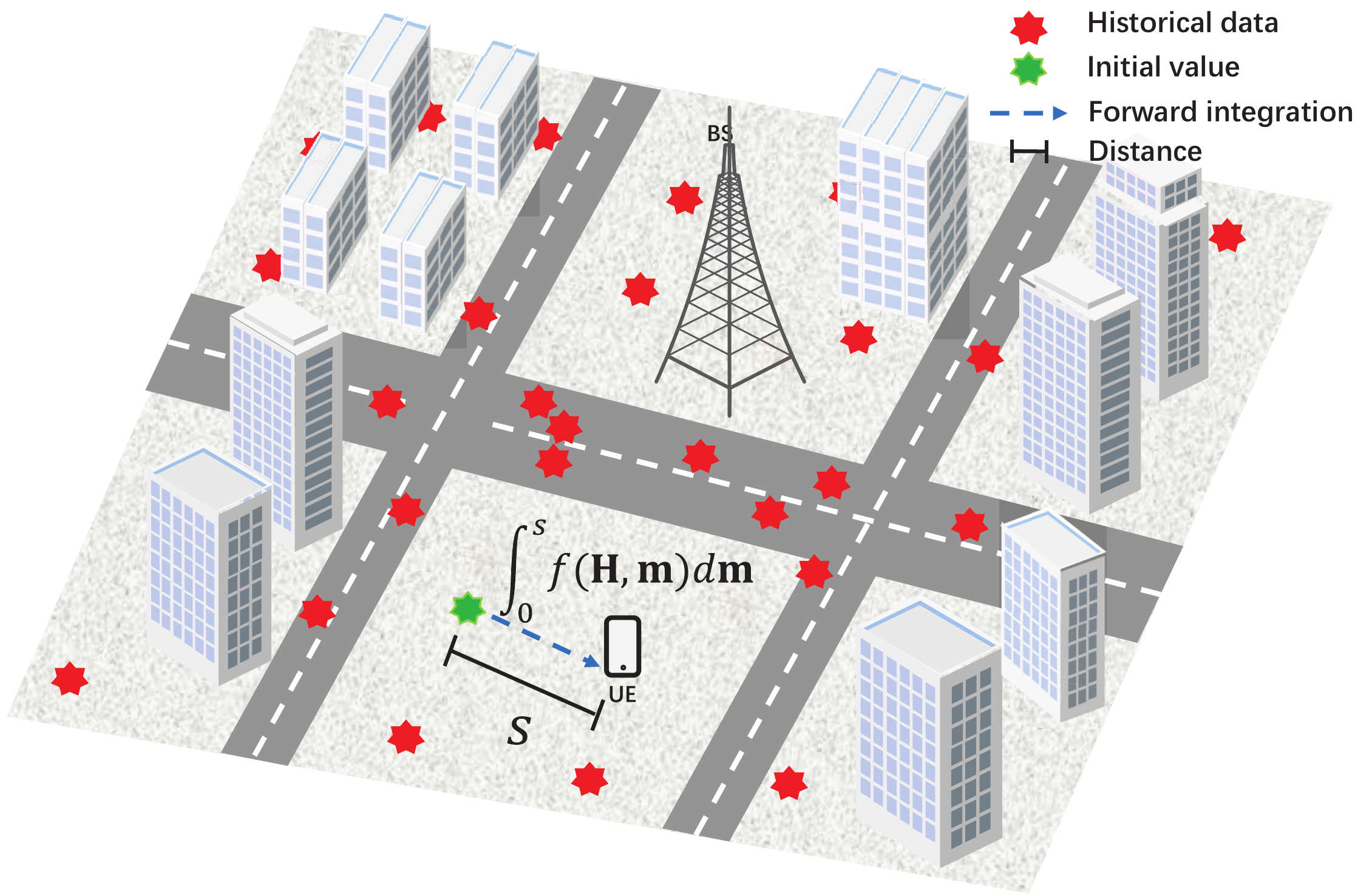}
	\vspace{-.5em}
	\caption{The process for obtaining static channel with forward integration through Neural ODE.}
	\vspace{-.5em}
	\label{integration}
\end{figure}

\subsection{Algorithm for Elimination and Compensation of the Doppler Shift with Angular-Delay Domain Channel}

Doppler shift compensation is the key step for obtaining the target mobile CSI with the static channel predicted by Neural ODE. Doppler compensation in existing works relies on AOA separation, which is quite complex and requires a lot of computational costs. However, since we already know the velocity vector, the Doppler compensation could be much easier, especially with the angular-delay domain channel.

As shown in Fig. \ref{Doppler}, the Doppler phase shift is only related to the velocity vector of the moving UE and the angle of arrival of the certain propagation path. Conveniently, as we have mentioned above, each element in the angular-delay domain channel matrix can be viewed approximately as an equivalent propagation path response ${g_{z ,q }}$ with known AOA and time delay shown in Equation \eqref{eq1}, \eqref{eq2}. Thus, the channel with Doppler shift ${\bf{G}}_{mobile}$ can be calculated
\begin{equation}
{{\bf{G}}_{mobile}} = {\bf{G}} \circ {\bf{D}},
\end{equation}
where $ \circ $ denotes the Hadamard product and ${\bf{D}} \in {N_t} \times {N_c}$ is the Doppler compensation matrix 
\begin{equation}
	{\bf{D}}(z,q){ = ^{ - j2\pi {{{v_u}} \over \lambda }\cos ({\theta _v} - {\theta _z} + \varphi ){\tau _q}}}.
\end{equation}
Dually, for Doppler elimination, we have
\begin{equation}
{\bf{G}} = {{\bf{G}}_{mobile}} \circ {\bf{E}},
\end{equation}
where ${\bf{E}} \in {N_t} \times {N_c}$ is the Doppler elimination matrix 
\begin{equation}
{\bf{E}}(z,q) = {e^{j2\pi {{{v_u}} \over \lambda }\cos ({\theta _v} - {\theta _z} + \varphi ){\tau _q}}},
\end{equation}

\subsection{The Iteration Algorithm for Sequential Motion Information Obtaining and Positioning}

In order to obtain the position and velocity of the UE so as to further obtain the CSI to be predicted, we propose a novel learning-based positioning network and an iterative motion extraction method. The overall target is to make use of the channel sequence estimated in a past time period and predict the position of the UE in the next time slot with its velocity vector. 
\begin{figure}
	\centering
	\includegraphics[width=0.35\textwidth]{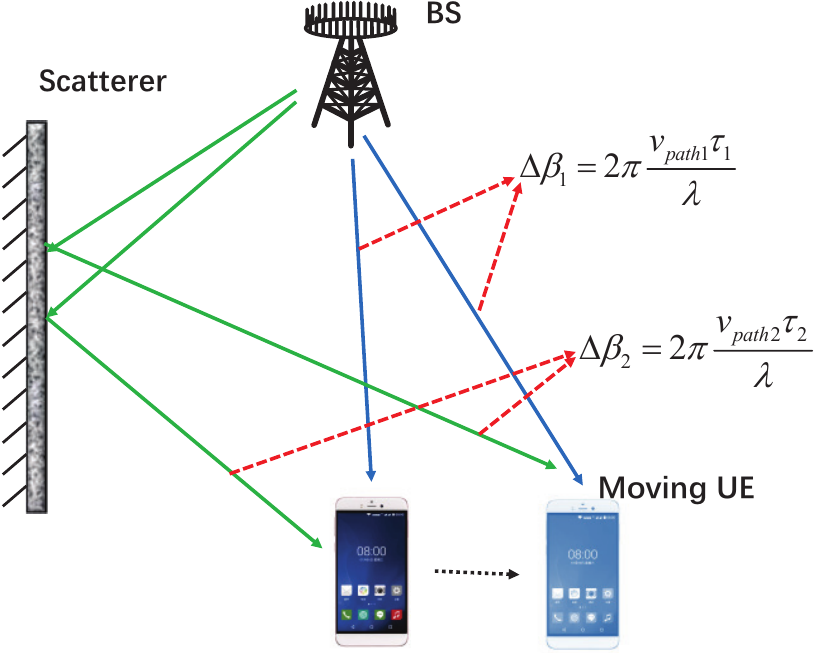}
	\vspace{-.5em}
	\caption{Schematic diagram of Doppler phase shift.}
	\vspace{-.5em}
	\label{Doppler}
\end{figure}

\begin{figure}
	\centering
	\includegraphics[width=0.45\textwidth]{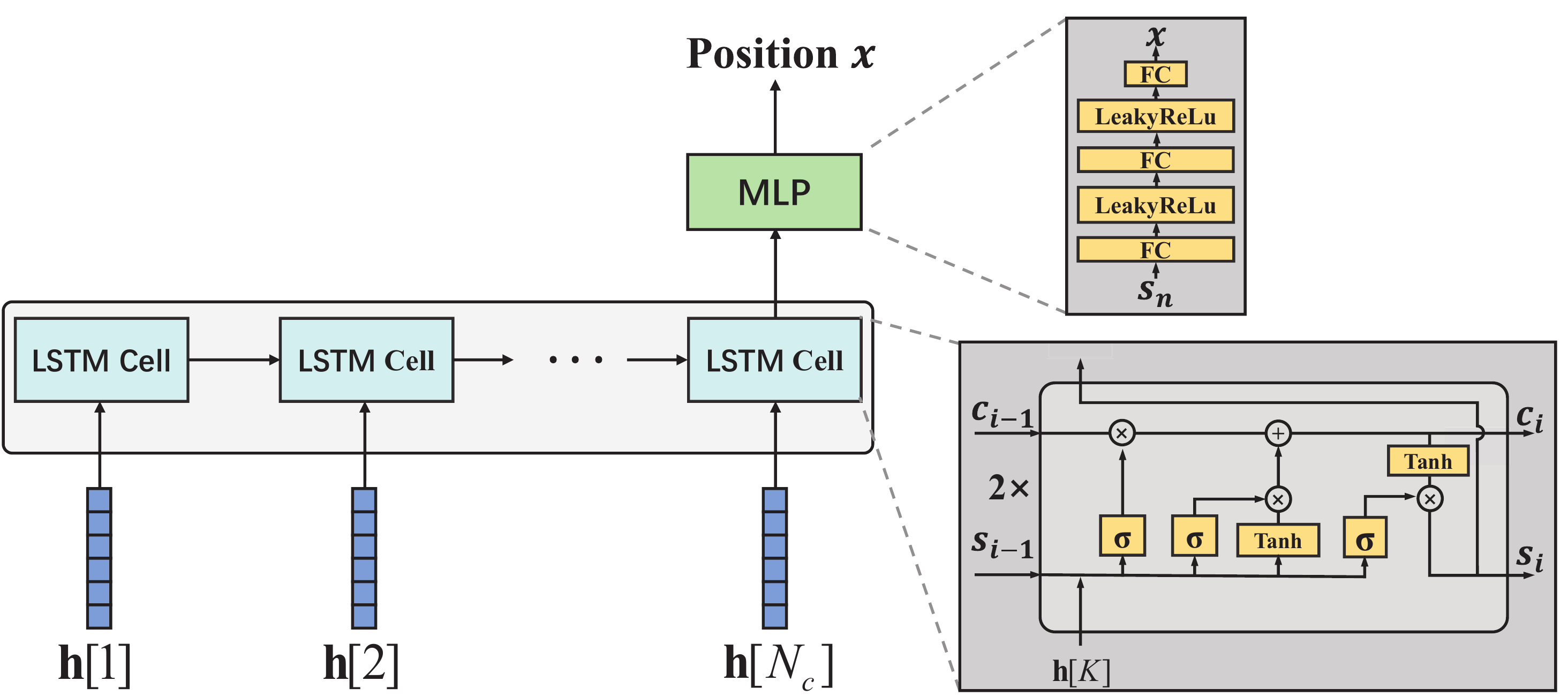}
	\vspace{-.5em}
	\caption{The LSTM structure for channel positioning network.}
	\vspace{-.5em}
	\label{LSTM}
\end{figure}

Positioning through a learning-based method is essentially a feature extraction process, and existing works widely adopt convolution structures \cite{cnn_space_frequency,cnn_angle_delay1,cnn_angle_delay2,ad_cnn3}. However, the positioning accuracy of those works is highly limited. The main reason is that there are several essential differences between CSI matrix elements and pixels in an image, which cause mismatches between the CSI matrix and convolution algorithm. First of all, it is the relationship between pixels rather than the coordinates of the pixels that really contain the critical information of images. This property makes the convolution and pooling operators fit the image data well due to the translational invariance. In contrast, the position index in the CSI matrix has clear correspondence with the index of antennas and carriers and is neither swappable nor translational. Thus, lots of crucial information losses when it flows through the convolutional neural network. On the other hand, images have a high smoothness characteristic since the values of adjacent pixels are usually similar. In contrast, the adjacent elements of the CSI matrix are generally quite numerically different. Such high frequency and unsmooth features make the convolution operator’s inherent smoothing effect which is suitable for images, however, become the bottleneck of the MIMO channel data representation.

\begin{figure}
	\centering
	\includegraphics[width=0.43\textwidth]{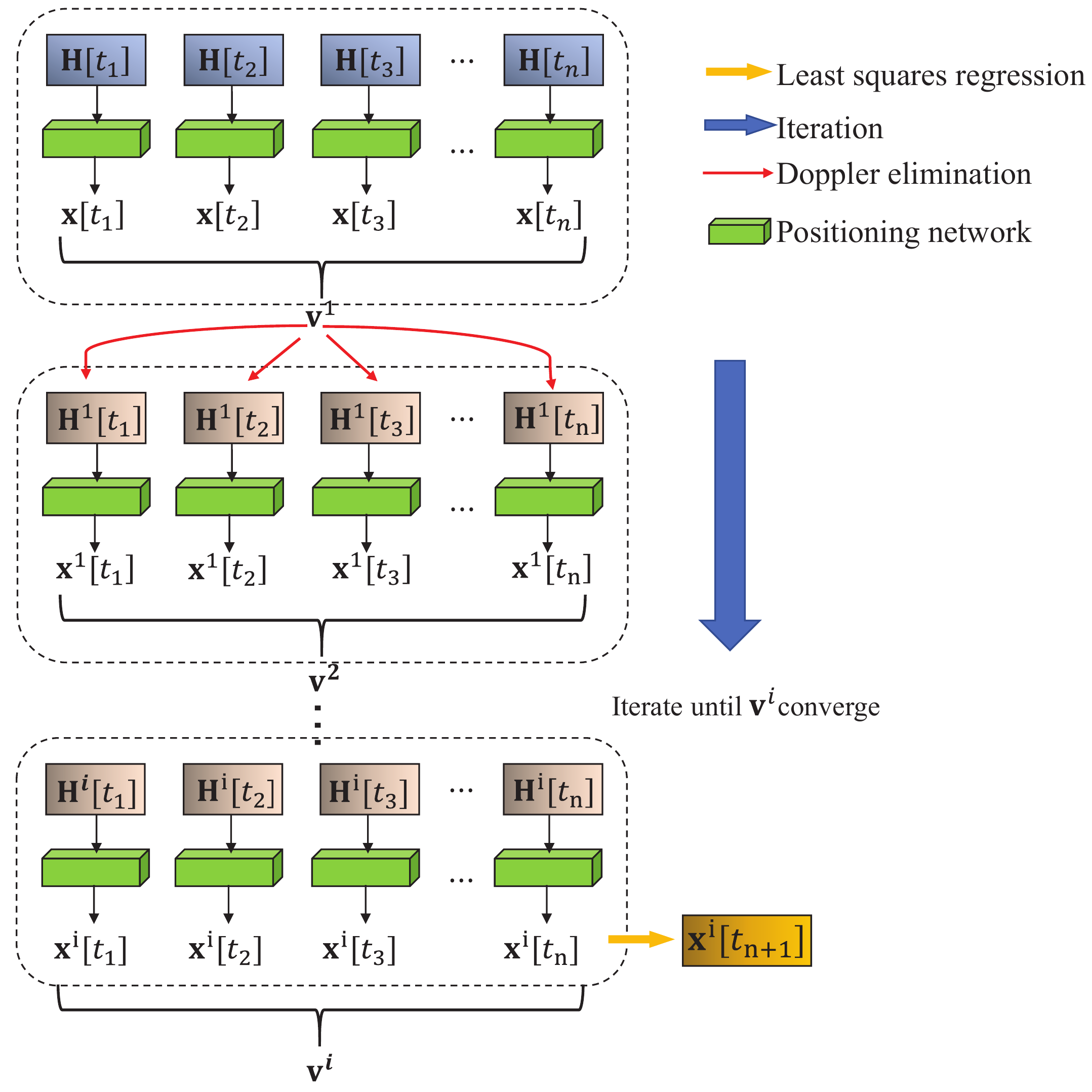}
	\vspace{-.5em}
	\caption{The iteration process for motion vector obtaining and positioning.}
	\vspace{-.5em}
	\label{iteration}
\end{figure}
As proven in \cite{10061451}, considering that the interval of adjacent subcarriers is strictly equal, the MIMO channel frequency response matrix has some uniform sequential characteristics in the spatial domain. Thus, the CSI can be viewed as an information sequence, and the sequential learning structure is more suitable for feature extraction. As shown in Fig. \ref{LSTM}, the network is composed of $N_c$ standard LSTM cells. The channel response vector of different sub-carriers are inputted into those cells respectively, and the targeted position is outputted in the last cell. Since it is not the whole channel matrix but one single column inputted to each cell, the parameters scale for each cell can be highly limited and ensure network robustness. 

Since the positioning network is trained with historic static channel data at the training state while the input is the measured channel under movement at the inferring stage, the output result will carry some inherent errors. So, we propose an iterative method to eliminate the Doppler effect in the estimated mobile channel sequence. Practically, the channel sampling interval is in the time scale of milliseconds, which is negligible compared with the time scale of user speed change. So, we can approximately assume that the user is moving uniformly in a straight line during this period of time. This assumption provides strong prior knowledge and further clarifies the iteration goal for the algorithm. Algorithm \ref{algorithm2} describes the detailed flow, and the key designing ideas can be summarized as follows. With the channel sequence $\{ {\bf{H}}[{t_1}],{\bf{H}}[{t_2}], \cdot  \cdot  \cdot ,{\bf{H}}[{t_n}]\} $, we can obtain a series of UE's position $\{ {\bf{x}}[{t_1}],{\bf{x}}[{t_n}], \cdot  \cdot  \cdot ,{\bf{x}}[{t_3}]\} $ through the channel positioning network. Since all the time points $\{ {t_1},{t_1}, \cdot  \cdot  \cdot ,{t_n}\} $ are known, the least squares regression algorithm can be adopted to find a linear equation to fit the sequence 
\begin{equation}
	{\bf{x}} = {{\bf{v}}^i}t + {\bf{\sigma }},
\end{equation}
where ${{\bf{v}}^i}$ is the velocity vector calculated at the $i$th iteration and ${\bf{\sigma }}$ is a constant vector.
Then, the Doppler elimination algorithm introduced above can be used to obtain the targeted static channel $\{ {{\bf{H}}^i}[{t_1}],{{\bf{H}}^i}[{t_2}], \cdot  \cdot  \cdot ,{{\bf{H}}^i}[{t_n}]\} $. Furthur, predicted position $\{ {{\bf{x}}^i}[{t_1}],{{\bf{x}}^i}[{t_2}], \cdot  \cdot  \cdot ,{{\bf{x}}^i}[{t_3}]\} $ as well as the velocity vector ${{\bf{v}}^i}$ can be obtained. Repeat the above process until $|{{\bf{v}}^{i + 1}} - {{\bf{v}}^i}| < \eta $ and finally the position ${{\bf{x}}^{i + 1}}[{t_{n + 1}}]$ to be predicted can be obtained. It should be noticed that since the positioning accuracy of the positioning network is relatively high, even with a mobile channel as input, the convergence of the iteration algorithm can be highly ensured.

\begin{algorithm}
	\small 
	\DontPrintSemicolon
	\SetAlgoLined
	\KwIn {Channel sequence obtained in a past period of time $\{ {\bf{H}}[{t_1}],{\bf{H}}[{t_2}], \cdot  \cdot  \cdot ,{\bf{H}}[{t_n}]\} $; The channel sampling time point $\{ {t_1},{t_2}, \cdot  \cdot  \cdot ,{t_n}\} $.
	}
	\KwOut {The UE's position in the next time slot ${{\bf{x}}^i}[{t_{n + 1}}]$; The UE's velocity vector ${{\bf{v}}^i}$.}
	Train positioning network with historical static CSI data\;
	Input $\{ {\bf{H}}[{t_1}],{\bf{H}}[{t_2}], \cdot  \cdot  \cdot ,{\bf{H}}[{t_n}]\} $ to the positioning network separately and get $\{ {\bf{x}}[{t_1}],{\bf{x}}[{t_2}], \cdot  \cdot  \cdot ,{\bf{x}}[{t_n}]\} $.\;
	Using the least squares regression algorithm to find the equation ${\bf{x}} = {\bf{v}}t + {\bf{\sigma }}$.\;
	${{\bf{v}}^0} \gets {\bf{v}}$; $|{{\bf{v}}^1}| \gets 0$; $i \gets 1$
	
	\While{$|{{\bf{v}}^i} - {{\bf{v}}^{i - 1}}| \ge \eta $}{
	Obtain $\{ {{\bf{H}}^i}[{t_1}],{{\bf{H}}^i}[{t_2}], \cdot  \cdot  \cdot ,{{\bf{H}}^i}[{t_n}]\} $ through Doppler elimination with ${{\bf{v}}^{i - 1}}$ and $\{ {{\bf{H}}^{i - 1}}[{t_1}],{{\bf{H}}^{i - 1}}[{t_2}], \cdot  \cdot  \cdot ,{{\bf{H}}^{i - 1}}[{t_n}]\} $.\;
	Input $\{ {{\bf{H}}^i}[{t_1}],{{\bf{H}}^i}[{t_2}], \cdot  \cdot  \cdot ,{{\bf{H}}^i}[{t_n}]\} $ to the positioning network separately and get $\{ {{\bf{x}}^i}[{t_1}],{{\bf{x}}^i}[{t_2}], \cdot  \cdot  \cdot ,{{\bf{x}}^i}[{t_n}]\} $.\;
	Using the least squares regression algorithm to find the equation ${{\bf{x}}^i} = {{\bf{v}}^i}t + {\bf{\sigma }}$.\;
	${i \gets i+1}$.\; 
	}
	Substitute ${t_{n + 1}}$ into the equation ${{\bf{x}}^i} = {{\bf{v}}^i}t + {\bf{\sigma }}$ to get ${\bf{x}}[{t_{n + 1}}]$.
	\caption{\textbf{Positioning and Motion Information Extraction}}
	\label{algorithm2}
\end{algorithm}

\section{Experimental Results and Analysis} \label{chap:results}
\subsection{Scene Setup and Datasets Generation} \label{scene}

As has been proven to be one of the most realistic deterministic channel models, the ray-tracing model is used to generate our dataset. Ray-tracing is a method for approximating the propagation of a wave in an environment using discrete rays \cite{686774,4685913,330158,8438326}. The discrete rays are traced by determining all possible specular images of TX/RX or by launching rays in different directions. The possible pathways and their corresponding interactions, such as reflections, diffraction, and diffuse scattering, are determined in both cases. 

As shown in Fig. \ref{model}, A practical outdoor scenario is chosen for setting up our 3D model. Here, we use Wireless Insight software from Remcom company to do the ray-tracing calculation. In Fig. \ref{model}, the height of Building 1 is 25m, Building 2 is 35m, Building 3, and Building 4 is 8m. The material for all the buildings is set as cement. Area 5 is a forest. Users are distributed in a 120m × 60m area. The central frequency is set as 3.5GHz. The BS is equipped with a ULA and is located 10m higher above building 2. The OFDM bandwidth is 100MHz, and the maximum number of paths is 25. Due to the fact that the channel sampling time interval is relatively small in a practical system, it can be assumed that the UE moves in a uniform linear way in any direction within the sequence time. Moreover, we calculate and apply the corresponding Doppler phase shift to each propagation path. 

\begin{figure}[htb!]
	\centering
	\includegraphics[width=0.35\textwidth]{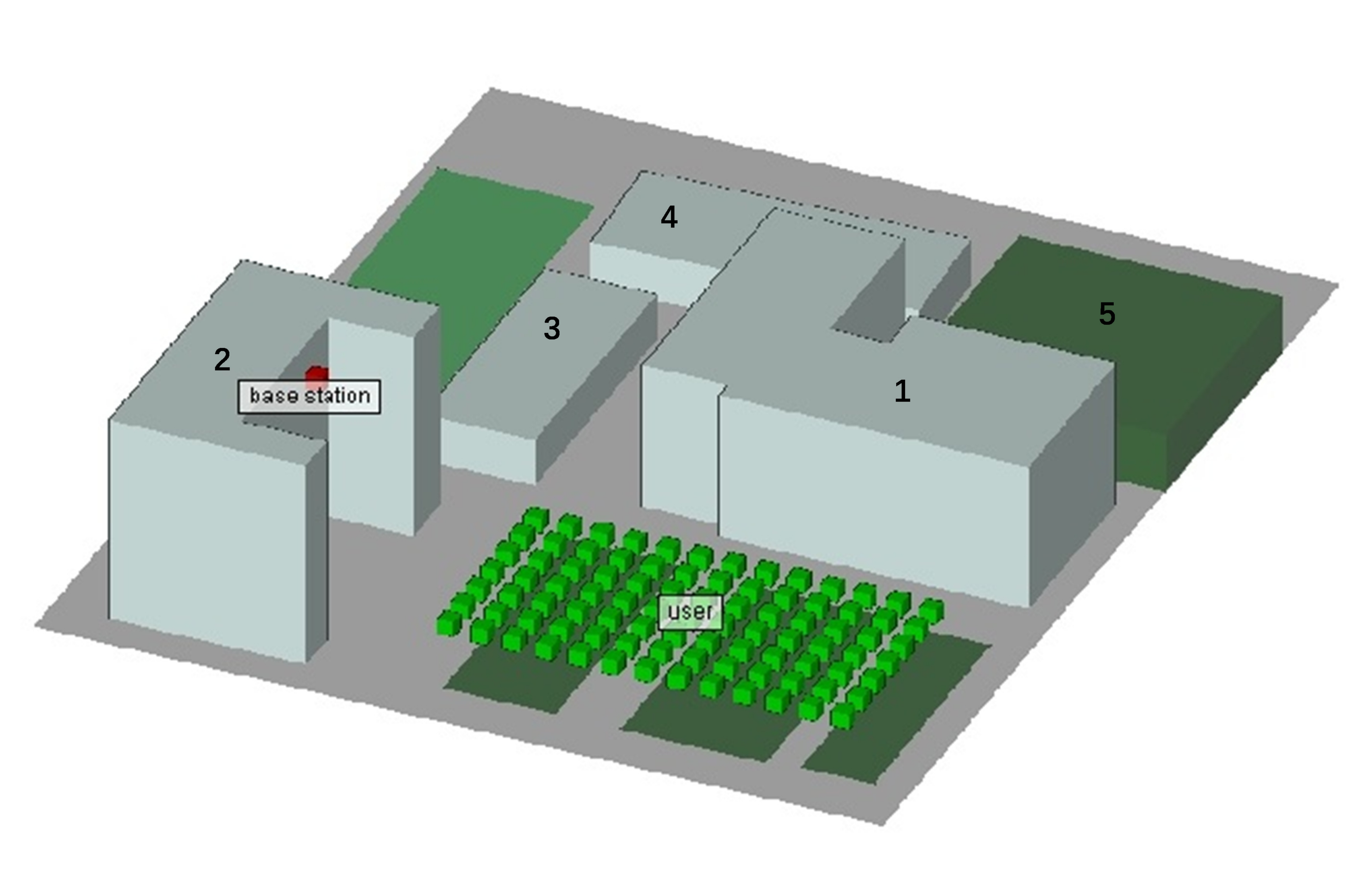}
	\vspace{-.5em}
	\caption{The 3D model of the Ray-tracing scene.}
	\label{model}
	\vspace{-.5em}
	%\vspace{-1cm}
\end{figure}

\subsection{Comparison Benchmarks} \label{Comparison}
Two state-of-the-art methods are selected as comparison benchmarks for comprehensively evaluating our proposed scheme. The performance measures include prediction accuracy, dataset size, and dependence on training data. 

First, we compare our scheme with the static channel database method in \cite{8968729}. It is convenient to directly compare our proposed method with it since they share the same historical static CSI database setting. In the comparing method, the author proposes that since the BS and main scatterers are static, certain wireless channel parameters mainly affected by static scatterers can also be assumed static, e.g., path loss, shadow fading, DOA, the number of paths, and time delay of main paths. Besides that, the velocity of BS is a dynamic parameter that can be obtained through the Global Positioning System (GPS) or other speed testing tools. Multiple Signal Classification (MUSIC) algorithm is used to extract channel parameters, and clustering methods are used assuming that the points in a cluster share the same static parameters. And then, other parameters, such as the path loss, are obtained through channel estimation. Here, for a more intuitive comparison, we relax the requirement of the comparing method and replace all parameters that need to be estimated with the ground truth. Then, we use the nearest neighbor interpolation to replace the original clustering method. With all the settings above, the predicted channel is supposed to be more accurate than the original comparing method. 

Then, we compare our scheme with the LSTM network, which is widely used as the benchmark in existing works adopting learning-based channel prediction methods \cite{9415201,9148836,9002073}. The LSTM network requires sequential data for training, while our proposed method only requires CSI of discrete points. So, to better illustrate the advantages of our proposed method in training data scale, we propose comparing our methods with the LSTM network while keeping the number of training data pieces equal in both methods where a sequence of data is defined as one piece of data. So, assuming the length of the sampling sequence is $z$, the scale of the data size of the comparing method is $z$ times larger than our proposed method. Because storing large amounts of data puts forward higher requirements for the base stations in future wireless communication systems, we want to show that our proposed method can achieve higher prediction accuracy with less training data through this comparison.
\begin{table}
	\caption{Main Parameters and Values for SCGNet.}
	\subtable[Scattering learning network]{
		\begin{minipage}[t]{1\linewidth}
		\begin{tabular}{ p{4cm}   p{3.7cm}}
	\toprule
	\textbf{Parameters} & \textbf{Value} \\
	\toprule
	Input dimension & 64×64×2\\
	
	Output dimension & 64×64×2 \\
	
	Activation function & Tanh \\
	
	Number of neurons in hidden & 256-768-512-256 \\			
	layer & \\
	Performance metric & Mean square error (MSE) \\
	Optimizer & Adam \\
	Training steps & $2 \times {10^5}$ \\
	Batch size & 20 \\
	Training samples & $80\% $ of the whole datasets \\
	\toprule
\end{tabular}
		\end{minipage}

	}
	\subtable[Direction embedding network]{
				\begin{minipage}[t]{1\linewidth}
		\begin{tabular}{ p{4cm}   p{3.7cm}}
			\toprule
			\textbf{Parameters} & \textbf{Value} \\
			\toprule
			Input dimension & 64×2\\
			
			Output dimension & 64×64 \\
			
			Activation function & Tanh \\
			
			Number of neurons in hidden & 512-256 \\	
			layer & \\		
			Performance metric & Mean square error (MSE) \\
			
			Optimizer & Adam \\
			
			Training steps & $2 \times {10^5}$ \\
			
			Batch size & 20 \\
			
			Training samples & $80\% $ of the whole datasets \\
			\toprule
		\end{tabular}
			\end{minipage}
	}
	\label{SCGnet}
\end{table}

\subsection{Model Settings and Complexity Analysing} \label{Setting}
The detailed model settings of SCGnet are shown in Tabel \ref{SCGnet}. The number of neurons of each hidden layer is selected from trial and error. Since the scattering learning network needs to learn the implicit scattering information, we design a relatively deep network with 256-768-512-256 hidden numbers for higher learning ability. The input and output dimension is ${N_t} \times {N_c} \times 2$. The direction embedding network serves an easier function. Thus, we only use an MLP with two hidden layers to reduce the network size and save computing costs. We adopt a two-layer LSTM structure for the positioning network with hidden numbers of 256-128. Tabel \ref{Positioning Network} shows the detailed model settings. 

For the SCGnet, the total number of parameters can be calculated as:
$(64 \times 64 \times 2 \times 256 + 256 + 256 \times 768 + 768 + 768 \times 512 + 512 + 512 \times 256 + 256 +256 \times 64 \times 64 \times2 + 64 \times 64\times2) + (64 \times 2 \times 512 + 512 + 512 \times 256 + 256 + 256 \times 64 \times 64 + 64 \times 64)= 6175232$. For the positioning network, the total number of parameters can be calculated as $[(256+64\times64\times2)\times256+256 + (256+128)\times128+128]\times4 + 128\times3 +3= 8849283$. Thus, the overall number of parameters is $6175232+8849283 = 15024515$. For the compared LSTM network, which adopts the structure of one hidden layer with 384 neurons, the total number of parameters can be calculated as $(64\times64\times2 + 384)\times384 + 384 + 384\times64\times64\times2 = 16328192$. Thus, the model complexity of our proposed network is approximately the same as the benchmark, which guarantees fairness in comparison. 

Mean square error (MSE) is chosen as the loss function, which can be written as
\begin{equation}
	{\rm{MSE}} = {1 \over i}\sum\limits_{w = 1}^i {{{({y_w} - \mathop {{y_w}}\limits^ \wedge  )}^2}}  ,
\end{equation}
where $y$ is the training label, ${\mathop y\limits^ \wedge}$ is the prediction value and $i$ is the dimension of output. It represents the average distance between the target value and the prediction value of all the output dimensions.

Normalized MSE (NMSE), which is an expectation value calculated across the testing dataset, is used to evaluate the prediction accuracy of the testing datasets since it is more convenient for comparison crossing different datasets. It can be mathematically written as
\begin{equation}
	\text{NMSE} = \mathbb{E}\left({{\sum\limits_{w = 1}^i {|{y_w} - {{\mathop {{y_w}|}\limits^ \wedge  }^2}} } \over {\sum\limits_{w = 1}^i {|{y_w}{|^2}} }}\right).
\end{equation}

For learning tasks, the performance is usually highly affected by the training dataset size. As we have mentioned above, the training dataset is generated by making use of the history records stored in the BS. So, we want to achieve favorable prediction performance while spending as little storage capacity of the BS as possible. Here, sampling density, defined as the average number of data samples in a unit area ($1m^2$), is used to evaluate the size of the training datasets. Three data sampling densities of $25$, $50$, and $100$ are considered in our experiment setting, and all the data are generated through a random process. 
We choose these three sampling densities because the experiment results will help better illustrate some essential characteristics of the network. It is noteworthy that the sampling density is much lower compared with other works in the field of AI-aided wireless communication with ray-tracing datasets. For example, 100000 channel samples are used in \cite{8395149} for a $40m \times 60m$ area, 121000 channel samples are used in \cite{9048929} for a $10m \times 10m$ scenery, and 50000 samples are used in \cite{8645463} for a $20m \times 30m$ scenery.

\begin{table}
	\caption{Main Parameters and Values for Positioning Network}
	\begin{center}
		\begin{tabular}{ p{4cm}   p{3.7cm}}
			\toprule
			\textbf{Parameters} & \textbf{Value} \\
			\toprule
			Input dimension & 64×64×2\\
			
			Output dimension & 3 \\
			
			Activation function & Tanh \\
			
			Number of neurons in hidden layer & 256-128\\			
			
			Performance metric & Mean square error (MSE) \\
			
			Optimizer & Adam \\
			
			Training steps & $2 \times {10^5}$ \\
			
			Batch size & 20 \\
			
			Training samples & $80\% $ of the whole datasets \\
			\toprule
		\end{tabular}	
	\end{center}
	\label{Positioning Network}
\end{table}

\subsection{Experimental Results and Analysis}

\begin{figure}
	\begin{minipage}[t]{1\linewidth}
  \centering
%	\hspace{-0.5cm}
\includegraphics[width=0.8\textwidth]{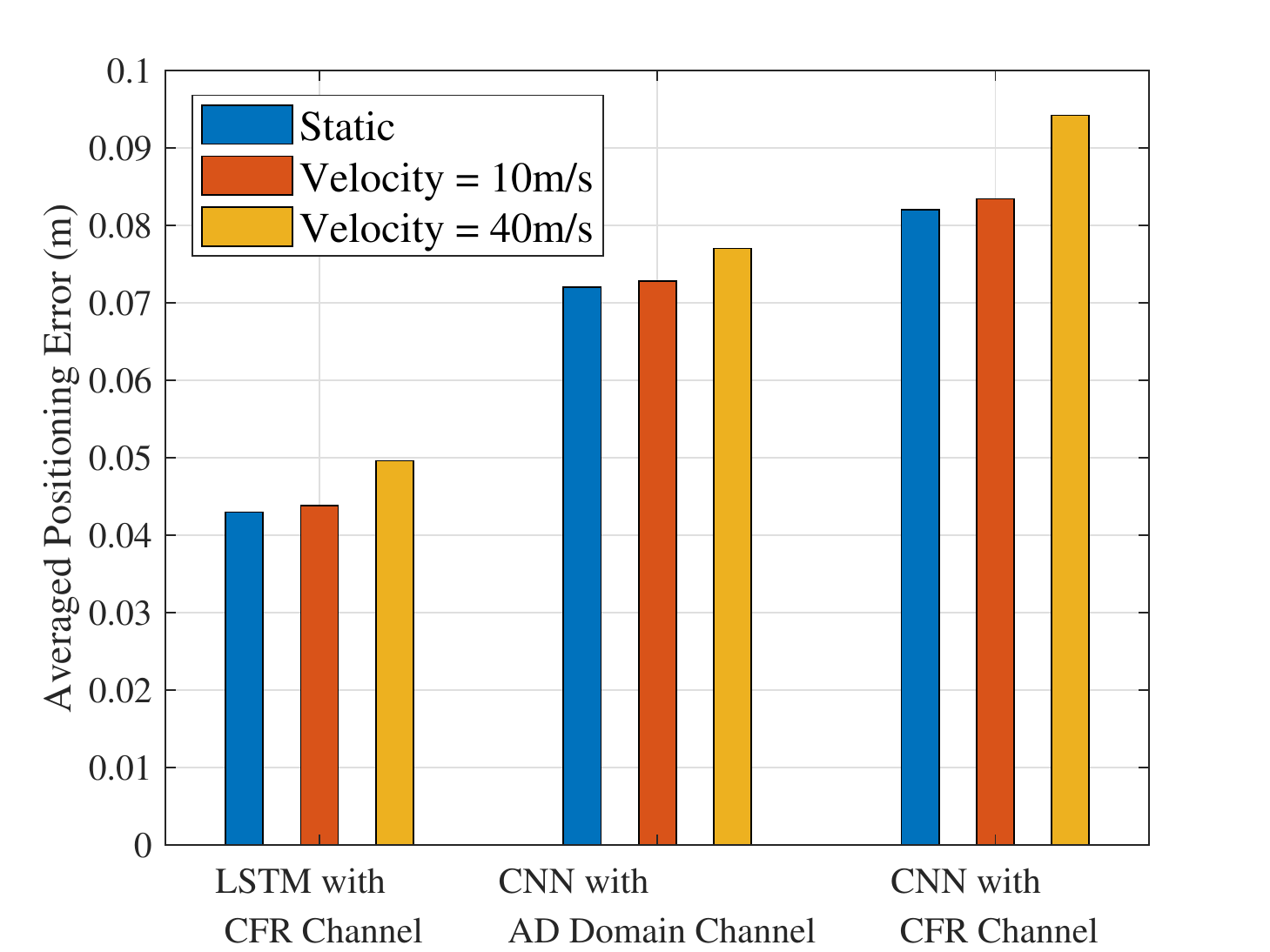}
\caption{Single point positioning accuracy of our proposed network and two comparing methods.}
\label{result1}
	\end{minipage}
	\begin{minipage}[t]{1\linewidth}
  \centering
%	\hspace{0.5cm}
\includegraphics[width=0.8\textwidth]{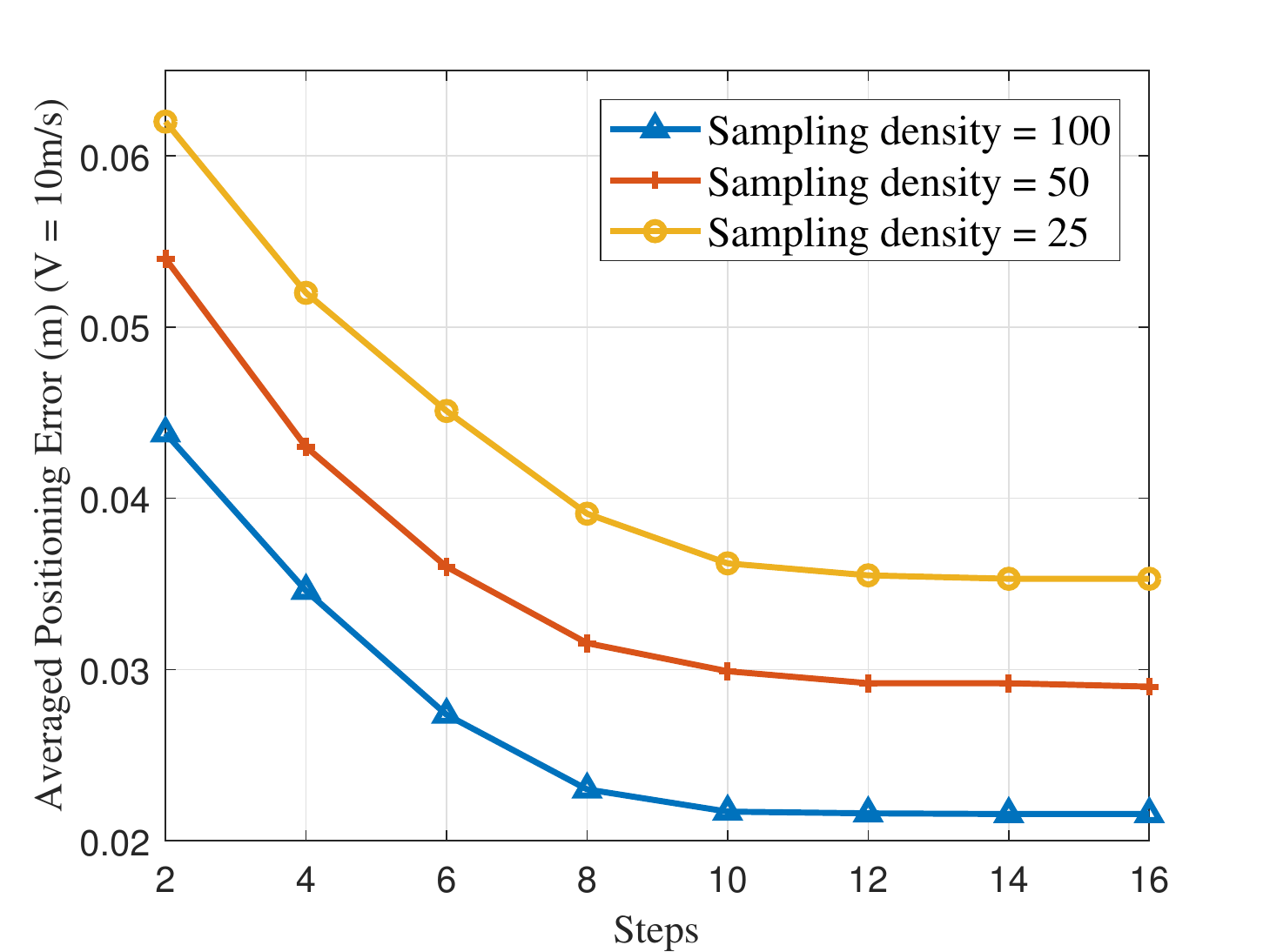}
\caption{Curves of positioning accuracy varies with the sequence length under different settings of settings of sampling density.}
\label{result2}
\end{minipage}
\end{figure}

Since CSI changes dramatically in space, positioning accuracy directly influences the final prediction accuracy. With the setting of sampling density of $100$, two start-of-the-arts are compared with our proposed network under different velocity settings. We keep the amounts parameters of the comparing network approximately the same to ensure the fairness of the comparison. As shown in Fig. \ref{result1}, the positioning network with CNN can achieve better performance with the angular-delay domain channel owing to the reason that channel in the angular-delay domain has some features similar to images. However, our proposed network can achieve higher accuracy. The reason is that we divide the channel matrix in the carrier, reducing the input dimension of each LSTM cell. Besides that, the network can extract channel features from the sequence perspective and reduce each individual cell's learning difficulty.

Fig. \ref{result2} shows the performance of our proposed iterative positioning method under three settings of sampling density. Here, we set the velocity of the UE as $40$m/s to show that our algorithm can handle positioning at high moving speed. It is obvious that using sequence CSI can achieve far higher accuracy than single CSI positioning. All three curves converge at steps = $10$, which means we only need the measured channel in the past $10$ time slots to reach the highest positioning accuracy. Results show that When the training data is sufficient, the error can be controlled to about one-fifth of the electromagnetic wavelength. Thus, the iterative algorithm ensures that the input of the following Neural ODE is close to the ground truth. We believe two factors lead to the results. One is that more information is taken used. The other more important factor is that the prior information of UE's motion mode is well combined with the positioning algorithm.

Fig. \ref{nmse1} illustrates two kinds of comparison. The first is the comparison between the Neural ODE-based method with the state-of-the-art channel static method we mentioned above, and the other is the comparison between SCGnet and normal MLP. Also, the prediction NMSE under three sampling densities is compared to show the dependence of different methods on data volume. Here, we set the velocity of UE to $10$m/s, which is common in actual scenarios. It is worth mentioning that the comparing method of Neural ODE with MLP is simply done by replacing the SCGnet of the proposed method with normal MLP and keeping the total number of parameters approximately the same. Results show that adopting the Neural ODE learning structure can lead to excellent performance gain compared with the traditional method. One interesting result is that in the comparing method, the prediction accuracy suffers a great decrease when the sampling density switches from $50$ to $25$. However, in the Neural ODE based methods, the accuracy will not decrease significantly even at a lower sampling density, proving our proposed method's feasibility. Switching from normal MLP to the proposed SCGnet, the prediction accuracy is obviously improved under all three sampling density settings. This result is natural since instead of directly learning the complex mapping from CSI the gradient, our proposed network integrates the physics prior information into the network design and dramatically reduces the learning complexity of the network.

\begin{figure}[htb!]
	\centering
	\includegraphics[width=0.4\textwidth]{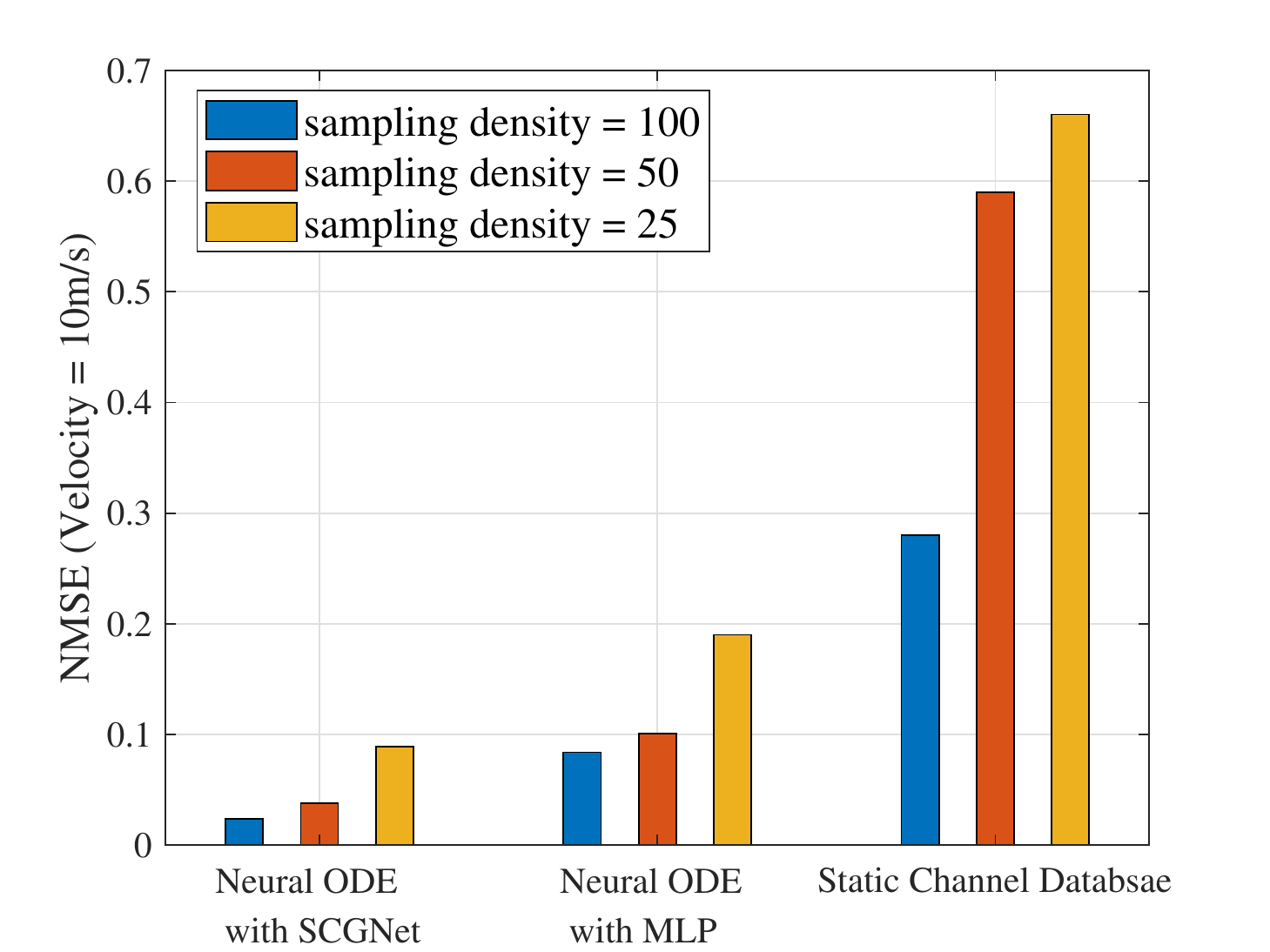}
	\vspace{-.5em}
	\caption{Prediction NMSE under different sampling density of three comparing methods.}
	\label{nmse1}
	\vspace{-.5em}
\end{figure}

\begin{figure}[htb!]
	\centering
	\subfigure[Sampling density = $50$]{\label{nmse21}
		\begin{minipage}[t]{0.75\linewidth}
			\centering
			\includegraphics[width=1\textwidth]{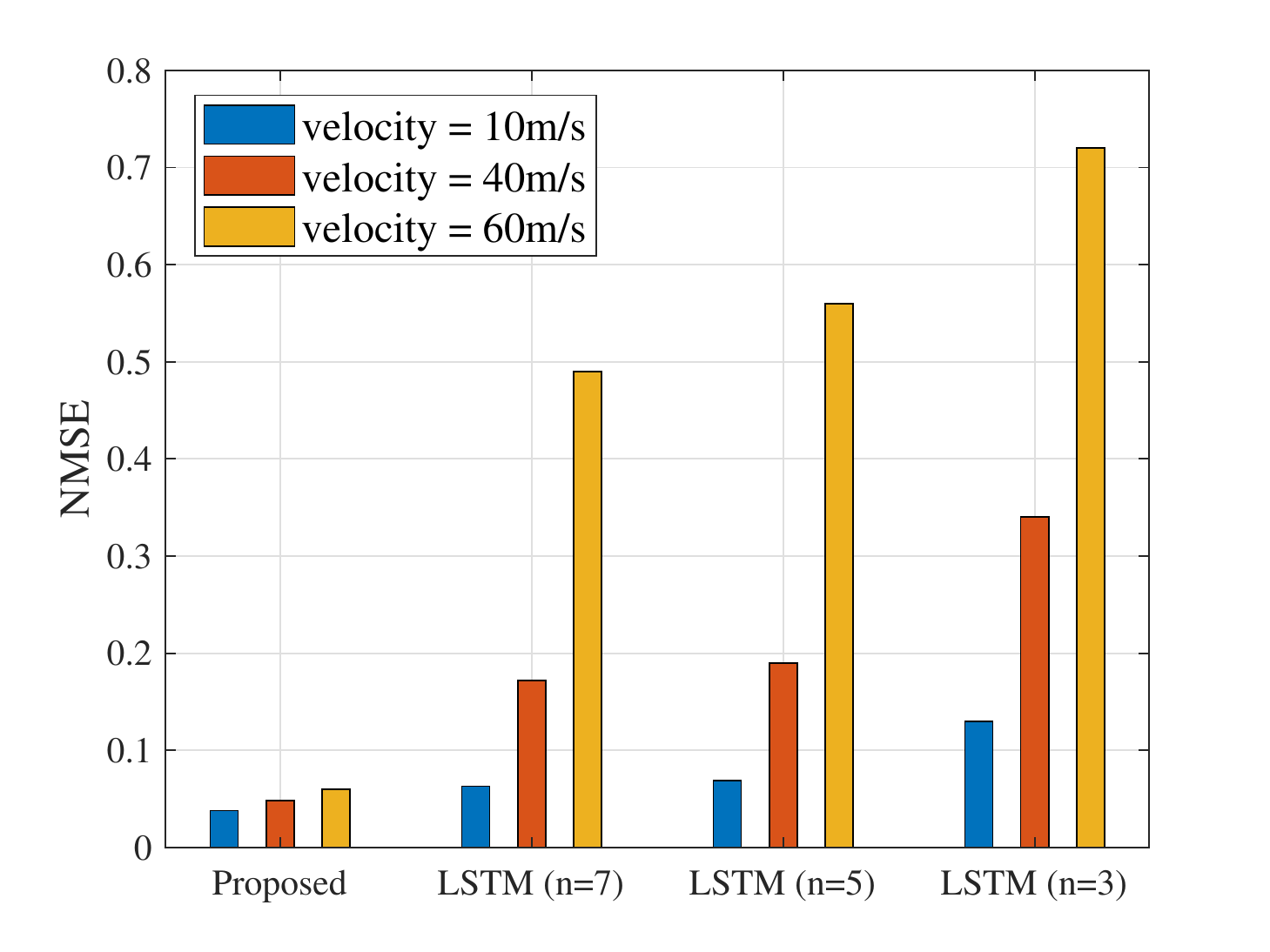}
		\end{minipage}

	}
	\subfigure[Sampling density = $100$] {\label{nmse22}
		\begin{minipage}[t]{0.75\linewidth}
			\centering
		\includegraphics[width=1\textwidth]{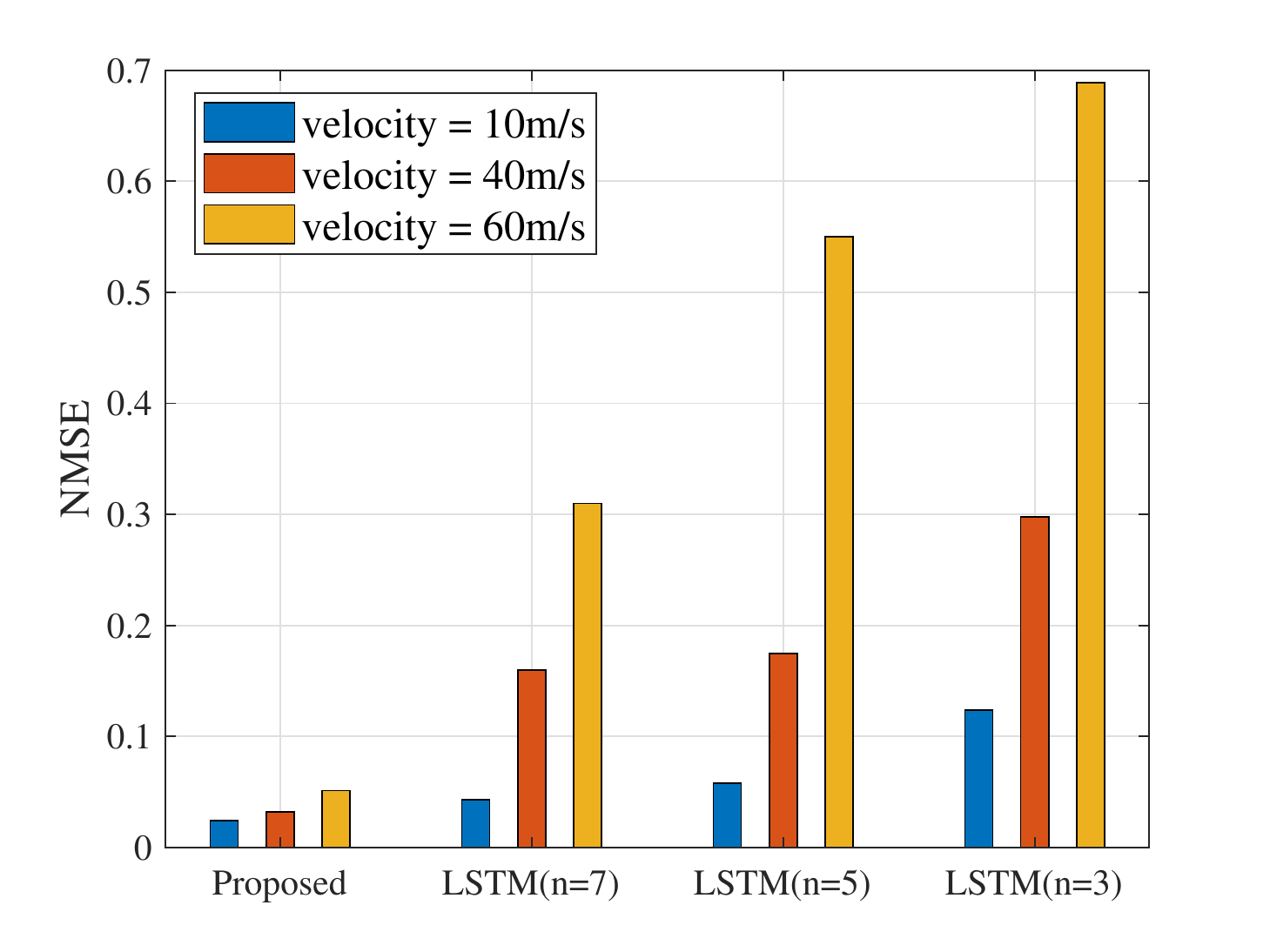}
		\end{minipage}
	}
	\centering
	\caption{Prediction NMSE comparison between proposed method and LSTM of different lenghths under different settings of UE's velocity.}
	\label{nmse2}
	\vspace{-.5em}
\end{figure}

In order to make a more comprehensive comparison between our proposed method and the LSTM network, three different lengths are adopted. Besides that, different UE's velocities and sampling densities are tested. The sequence interval of the training dataset for LSTM is $1$ms, and both methods output channel matrix in the angular-delay domain. Since the LSTM network with a small training dataset is difficult to converge, we only make the comparison with the sampling density of $50$ and $100$. The overall results are shown in Fig. \ref{nmse2}. It can be seen that with the increase in sequence length, the prediction accuracy can be improved for the LSTM network. However, the prediction accuracy of length = $7$ and length = $5$ are very close, especially in the low-speed scenario. So, we believe the prediction accuracy is close to optimal when the length = $7$. It can be seen that,  benefiting from the targeted design for unique physics processes, our proposed method can achieve much higher accuracy than the benchmark in any scenario. Besides that, it can be seen from the results that the sequence-based network is quite sensitive to the velocity of UE, and the performance will decrease tremendously when the speed increases. The main reason behind this is that the performance of these kinds of networks highly relies on the correlation between sequential data. With the increase of UE's moving speed, the spatial spacing between sequence sampling points will also increase, weakening spatial correlation between sequences. On the contrary, our proposed method can maintain high prediction accuracy even with UE's fast movement.

\begin{figure}
	\centering
	\includegraphics[width=0.35\textwidth]{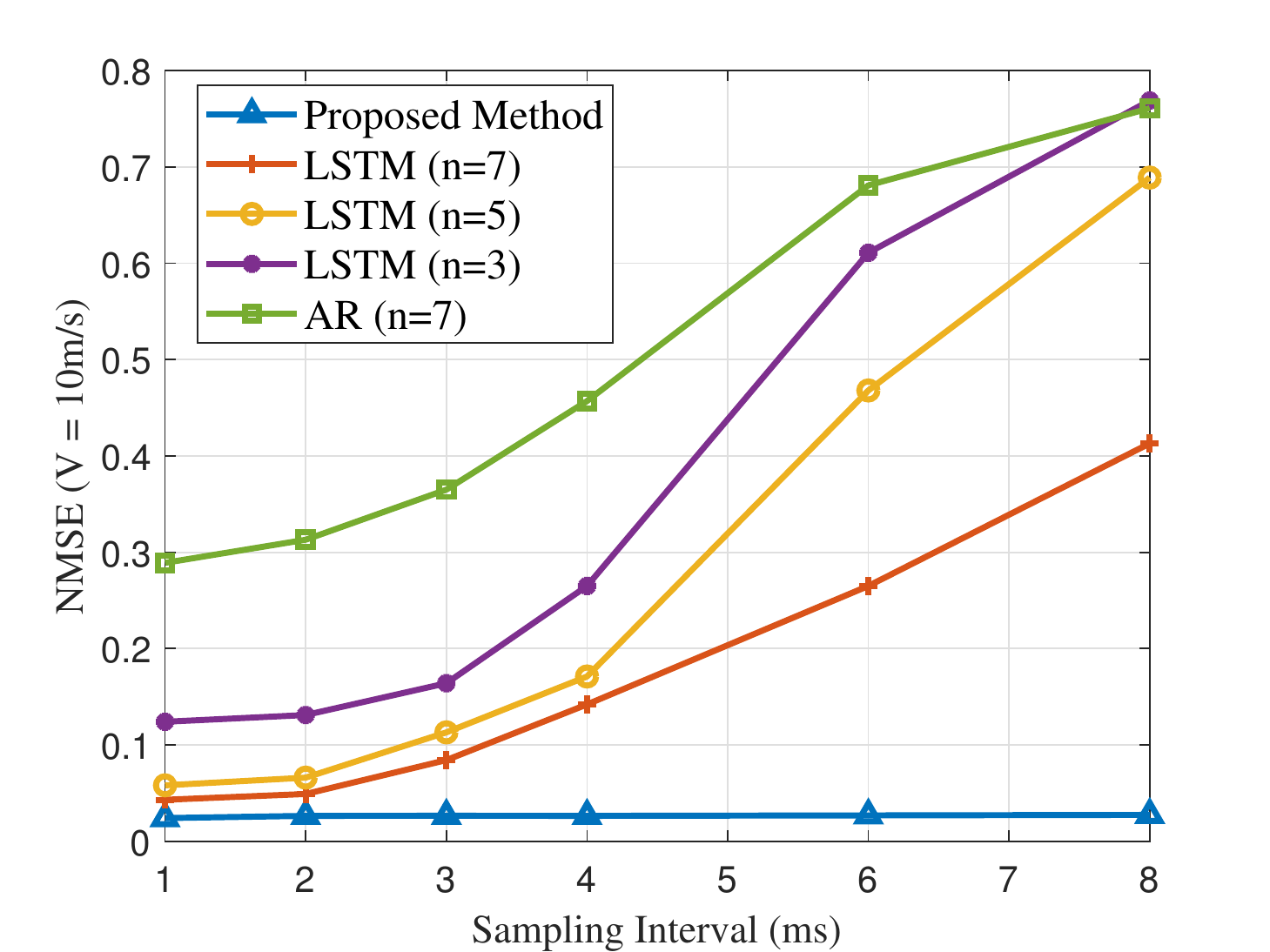}
	\caption{Prediction NMSE under different sampling intervals of the three comparing methods.}
	\label{nmse4}
\end{figure}

In order to better illustrate the advantage of our proposed methods over sequential learning structure, Fig. \ref{nmse4} compares
the prediction performance of the two kinds of learning-based approaches with greater
spatial spacing between adjacent sampling points by increasing the sampling time interval. Besides that, the traditional AR model with length = $7$ is added as a benchmark. It is worth noting that such an experiment is of practical significance since reducing sampling intervals can help save pilots needed in channel estimation. As can be seen from the results, although still better than the AR model, the prediction of LSTM networks dropped sharply with the increasing sampling interval. For the case of length = $3$ and length = $5$, the networks are even harder to converge with longer intervals. However, the sampling interval has no influence on the prediction accuracy of our method in an ideal scenario where no velocity of UE changes happened in this period. Therefore, our scheme is more feasible in practical application.

\section{CONCLUSION} \label{chap:conclusion}
In this paper, we addressed the practical challenge of efficiently and accurately mobile MIMO channel prediction with only a set of channel instances obtained within a certain wireless communication environment. We proved that the spatial channel gradient is only decided by the CSI itself and the spatial direction. Therefore, the channel prediction problem was modeled as an ordinary differential equation problem with a known initial value. Innovatively, we proposed a physics-inspired SCGnet to implicitly represent the channel gradient with respect to spatial displacement at any location and adopt the Neural ODE learning structure to update the parameters and do the forward calculation. Further, to obtain the integral length and direction of the Neural ODE, we proposed a novel positioning network and an iterative algorithm for sequential positioning. Experiments have shown the effectiveness of our proposed learning scheme. The values predicted through our method were very close to the ground truth generated through the Ray-Tracing model.

It is an encouraging achievement for us to open a new prospect of switching from the past way of adopting a purely data-driven universal learning architecture to design a physics-inspired learning structure aimed at unique data characteristics or tasks in wireless communication. In fact, different from many other machine learning tasks, there is a lot of unique prior knowledge in wireless communication that has not been fully used in existing works. Our work is a good example to illustrate the detailed thoughts of fusing prior information into the design of a network structure or learning scheme. We believe that our work has the potential to guide the learning system design of sensing wireless communication scenes and intelligent wireless communication in the future.

\bibliographystyle{IEEEtran}
\bibliography{bibfile}

\begin{thebibliography}{10}
\providecommand{\url}[1]{#1}
\csname url@samestyle\endcsname
\providecommand{\newblock}{\relax}
\providecommand{\bibinfo}[2]{#2}
\providecommand{\BIBentrySTDinterwordspacing}{\spaceskip=0pt\relax}
\providecommand{\BIBentryALTinterwordstretchfactor}{4}
\providecommand{\BIBentryALTinterwordspacing}{\spaceskip=\fontdimen2\font plus
\BIBentryALTinterwordstretchfactor\fontdimen3\font minus
  \fontdimen4\font\relax}
\providecommand{\BIBforeignlanguage}[2]{{%
\expandafter\ifx\csname l@#1\endcsname\relax
\typeout{** WARNING: IEEEtran.bst: No hyphenation pattern has been}%
\typeout{** loaded for the language `#1'. Using the pattern for}%
\typeout{** the default language instead.}%
\else
\language=\csname l@#1\endcsname
\fi
#2}}
\providecommand{\BIBdecl}{\relax}
\BIBdecl

\bibitem{9427230}
C.~Wu, X.~Yi, Y.~Zhu, W.~Wang, L.~You, and X.~Gao, ``Channel prediction in
  high-mobility massive {MIMO}: From spatio-temporal autoregression to deep
  learning,'' \emph{IEEE Journal on Selected Areas in Communications}, vol.~39,
  no.~7, pp. 1915--1930, 2021.

\bibitem{9839184}
W.~Li, H.~Yin, and M.~Debbah, ``A super-resolution channel prediction approach
  based on extended matrix pencil method,'' in \emph{ICC 2022 - IEEE
  International Conference on Communications}, 2022, pp. 1355--1360.

\bibitem{8395053}
C.~Luo, J.~Ji, Q.~Wang, X.~Chen, and P.~Li, ``Channel state information
  prediction for {5G} wireless communications: A deep learning approach,''
  \emph{IEEE Transactions on Network Science and Engineering}, vol.~7, no.~1,
  pp. 227--236, 2020.

\bibitem{9625179}
Z.~Xiao, Z.~Zhang, C.~Huang, Q.~Yang, and X.~Chen, ``Channel prediction based
  on a novel physics-inspired generative learning structure,'' in \emph{2021
  IEEE 94th Vehicular Technology Conference (VTC2021-Fall)}, 2021, pp. 01--06.

\bibitem{4536852}
Y.~Zhang, S.~Liu, Y.~Rui, S.~Zhou, and J.~Wang, ``Channel prediction assisted
  by radio propagation environments information,'' in \emph{2008 4th IEEE
  International Conference on Circuits and Systems for Communications}, 2008,
  pp. 733--736.

\bibitem{9791407}
Z.~Xiao, Z.~Zhang, C.~Huang, X.~Chen, C.~Zhong, and M.~Debbah, ``{C-GRBFnet}: A
  physics-inspired generative deep neural network for channel representation
  and prediction,'' \emph{IEEE Journal on Selected Areas in Communications},
  vol.~40, no.~8, pp. 2282--2299, 2022.

\bibitem{4735375}
Y.~Zhang, S.~Liu, Y.~Rui, S.~Zhou, and J.~Wang, ``Wideband wireless channel
  predictor relying on environment information,'' in \emph{2008 8th
  International Symposium on Antennas, Propagation and EM Theory}, 2008, pp.
  945--948.

\bibitem{8932272}
X.~Ma, J.~Zhang, Y.~Zhang, and Z.~Ma, ``Data scheme-based wireless channel
  modeling method: Motivation, principle and performance,'' \emph{Journal of
  Communications and Information Networks}, vol.~2, no.~3, pp. 41--51, 2017.

\bibitem{8116491}
X.~Ma, J.~Zhang, Y.~Zhang, Z.~Ma, and Y.~Zhang, ``A {PCA}-based modeling method
  for wireless {MIMO} channel,'' in \emph{2017 IEEE Conference on Computer
  Communications Workshops (INFOCOM WKSHPS)}, 2017, pp. 874--879.

\bibitem{2016OnThe}
M.~Uccellari, F.~Facchini, M.~Sola, E.~Sirignano, and S.~Tondelli, ``On the
  application of support vector machines to the prediction of propagation
  losses at 169 mhz for smart metering applications,'' \emph{IET Microwaves
  Antennas and Propagation}, vol.~12, no.~3, 2016.

\bibitem{554747}
P.-R. Chang and W.-H. Yang, ``Environment-adaptation mobile radio propagation
  prediction using radial basis function neural networks,'' \emph{IEEE
  Transactions on Vehicular Technology}, vol.~46, no.~1, pp. 155--160, 1997.

\bibitem{6748900}
L.~Azpilicueta, M.~Rawat, K.~Rawat, F.~M. Ghannouchi, and F.~Falcone, ``A ray
  launching-neural network approach for radio wave propagation analysis in
  complex indoor environments,'' \emph{IEEE Transactions on Antennas and
  Propagation}, vol.~62, no.~5, pp. 2777--2786, 2014.

\bibitem{7590098}
G.~P. Ferreira, L.~J. Matos, and J.~M.~M. Silva, ``Improvement of outdoor
  signal strength prediction in {UHF} band by artificial neural network,''
  \emph{IEEE Transactions on Antennas and Propagation}, vol.~64, no.~12, pp.
  5404--5410, 2016.

\bibitem{6945858}
R.~O. Adeogun, P.~D. Teal, and P.~A. Dmochowski, ``Extrapolation of {MIMO}
  mobile-to-mobile wireless channels using parametric-model-based prediction,''
  \emph{IEEE Transactions on Vehicular Technology}, vol.~64, no.~10, pp.
  4487--4498, 2015.

\bibitem{841729}
A.~Duel-Hallen, S.~Hu, and H.~Hallen, ``Long-range prediction of fading
  signals,'' \emph{IEEE Signal Processing Magazine}, vol.~17, no.~3, pp.
  62--75, 2000.

\bibitem{2002Recurrent}
J.~T. Connor, R.~D. Martin, and L.~E. Atlas, ``Recurrent neural networks and
  robust time series prediction,'' \emph{IEEE Transactions on Neural Networks},
  vol.~5, no.~2, pp. 240--254, 2002.

\bibitem{6755477}
T.~Ding and A.~Hirose, ``Fading channel prediction based on combination of
  complex-valued neural networks and chirp {Z}-transform,'' \emph{IEEE
  Transactions on Neural Networks and Learning Systems}, vol.~25, no.~9, pp.
  1686--1695, 2014.

\bibitem{9128426}
W.~Jiang and H.~D. Schotten, ``Recurrent neural networks with long short-term
  memory for fading channel prediction,'' in \emph{2020 IEEE 91st Vehicular
  Technology Conference (VTC2020-Spring)}, 2020, pp. 1--5.

\bibitem{9277535}
Y.~Yang, F.~Gao, C.~Xing, J.~An, and A.~Alkhateeb, ``Deep multimodal learning:
  Merging sensory data for massive {MIMO} channel prediction,'' \emph{IEEE
  Journal on Selected Areas in Communications}, pp. 1--1, 2020.

\bibitem{9569281}
Z.~Xiao, Z.~Zhang, C.~Huang, C.~Zhong, and X.~Chen, ``{GPAE-LSTMnet}: A novel
  learning structure for mobile mimo channel prediction,'' in \emph{2021 IEEE
  32nd Annual International Symposium on Personal, Indoor and Mobile Radio
  Communications (PIMRC)}, 2021, pp. 1--6.

\bibitem{8979256}
J.~{Yuan}, H.~Q. {Ngo}, and M.~{Matthaiou}, ``Machine learning-based channel
  prediction in massive {MIMO} with channel aging,'' \emph{IEEE Transactions on
  Wireless Communications}, vol.~19, no.~5, pp. 2960--2973, 2020.

\bibitem{8904286}
Y.~Huangfu, J.~Wang, R.~Li, C.~Xu, X.~Wang, H.~Zhang, and J.~Wang, ``Predicting
  the mumble of wireless channel with sequence-to-sequence models,'' in
  \emph{2019 IEEE 30th Annual International Symposium on Personal, Indoor and
  Mobile Radio Communications (PIMRC)}, 2019, pp. 1--7.

\bibitem{9978073}
Z.~Xiao, Z.~Zhang, Z.~Chen, Z.~Yang, and R.~Jin, ``Mobile {MIMO} channel
  prediction with {ODE-RNN}: a physics-inspired adaptive approach,'' in
  \emph{2022 IEEE 33rd Annual International Symposium on Personal, Indoor and
  Mobile Radio Communications (PIMRC)}, 2022, pp. 1301--1307.

\bibitem{2018Optimization}
P.~Stapor, F.~Frhlich, and J.~Hasenauer, ``Optimization and uncertainty
  analysis of ode models using 2nd order adjoint sensitivity analysis,'' 2018.

\bibitem{2019Fingerprint}
X.~Sun, C.~Wu, X.~Gao, and G.~Y. Li, ``Fingerprint-based localization for
  massive {MIMO-OFDM} system with deep convolutional neural networks,''
  \emph{IEEE Transactions on Vehicular Technology}, vol.~PP, no.~99, pp. 1--1,
  2019.

\bibitem{2020Recurrent}
J.~Wei and H.~D. Schotten, ``Recurrent neural networks with long short-term
  memory for fading channel prediction,'' in \emph{2020 IEEE 91st Vehicular
  Technology Conference (VTC2020-Spring)}, 2020.

\bibitem{NEURIPS2018_69386f6b}
R.~T.~Q. Chen, Y.~Rubanova, J.~Bettencourt, and D.~K. Duvenaud, ``Neural
  ordinary differential equations,'' in \emph{Advances in Neural Information
  Processing Systems}, S.~Bengio, H.~Wallach, H.~Larochelle, K.~Grauman,
  N.~Cesa-Bianchi, and R.~Garnett, Eds., vol.~31.\hskip 1em plus 0.5em minus
  0.4em\relax Curran Associates, Inc., 2018.

\bibitem{2020Adaptive}
J.~Zhuang, N.~Dvornek, X.~Li, S.~Tatikonda, X.~Papademetris, and J.~Duncan,
  ``Adaptive checkpoint adjoint method for gradient estimation in neural
  {ODE},'' \emph{arXiv e-prints}, 2020.

\bibitem{8292280}
J.~Vieira, E.~Leitinger, M.~Sarajlic, X.~Li, and F.~Tufvesson, ``Deep
  convolutional neural networks for massive {MIMO} fingerprint-based
  positioning,'' in \emph{2017 IEEE 28th Annual International Symposium on
  Personal, Indoor, and Mobile Radio Communications (PIMRC)}, 2017, pp. 1--6.

\bibitem{cnn_space_frequency}
M.~Arnold, J.~Hoydis, and S.~t. Brink, ``Novel massive {MIMO} channel sounding
  data applied to deep learning-based indoor positioning,'' in \emph{SCC 2019;
  12th International ITG Conference on Systems, Communications and Coding},
  2019, pp. 1--6.

\bibitem{cnn_angle_delay1}
J.~Vieira, E.~Leitinger, M.~Sarajlic, X.~Li, and F.~Tufvesson, ``Deep
  convolutional neural networks for massive {MIMO} fingerprint-based
  positioning,'' in \emph{2017 IEEE 28th Annual International Symposium on
  Personal, Indoor, and Mobile Radio Communications (PIMRC)}, 2017, pp. 1--6.

\bibitem{cnn_angle_delay2}
X.~Sun, C.~Wu, X.~Gao, and G.~Y. Li, ``Deep convolutional neural networks
  enabled fingerprint localization for massive {MIMO-OFDM} system,'' in
  \emph{2019 IEEE Global Communications Conference (GLOBECOM)}, 2019, pp. 1--6.

\bibitem{ad_cnn3}
C.~Wu, X.~Yi, W.~Wang, Q.~Huang, and X.~Gao, ``{3D} {CNN}-enabled positioning
  in {3D} massive {MIMO-OFDM} systems,'' in \emph{ICC 2020 - 2020 IEEE
  International Conference on Communications (ICC)}, 2020, pp. 1--6.

\bibitem{10061451}
Z.~Chen, Z.~Zhang, Z.~Xiao, Z.~Yang, and K.-K. Wong, ``Viewing channel as
  sequence rather than image: A {2-D} {Seq2Seq} approach for efficient
  {MIMO-OFDM} {CSI} feedback,'' \emph{IEEE Transactions on Wireless
  Communications}, vol.~22, no.~11, pp. 7393--7407, 2023.

\bibitem{686774}
G.~Liang and H.~Bertoni, ``A new approach to {3-D} ray tracing for propagation
  prediction in cities,'' \emph{IEEE Transactions on Antennas and Propagation},
  vol.~46, no.~6, pp. 853--863, 1998.

\bibitem{4685913}
K.~Schuler, D.~Becker, and W.~Wiesbeck, ``Extraction of virtual scattering
  centers of vehicles by ray-tracing simulations,'' \emph{IEEE Transactions on
  Antennas and Propagation}, vol.~56, no.~11, pp. 3543--3551, 2008.

\bibitem{330158}
M.~Lawton and J.~McGeehan, ``The application of a deterministic ray launching
  algorithm for the prediction of radio channel characteristics in small-cell
  environments,'' \emph{IEEE Transactions on Vehicular Technology}, vol.~43,
  no.~4, pp. 955--969, 1994.

\bibitem{8438326}
D.~He, B.~Ai, K.~Guan, L.~Wang, Z.~Zhong, and T.~Kürner, ``The design and
  applications of high-performance ray-tracing simulation platform for {5G} and
  beyond wireless communications: A tutorial,'' \emph{IEEE Communications
  Surveys and Tutorials}, vol.~21, no.~1, pp. 10--27, 2019.

\bibitem{8968729}
H.~Li, Y.~Li, and S.~Zhou, ``Static {CSI} extraction and application in the
  tomographic channel model,'' \emph{China Communications}, vol.~16, no.~12,
  pp. 132--144, 2019.

\bibitem{9415201}
W.-S. Son and D.~S. Han, ``Analysis on the channel prediction accuracy of deep
  learning-based approach,'' in \emph{2021 International Conference on
  Artificial Intelligence in Information and Communication (ICAIIC)}, 2021, pp.
  140--143.

\bibitem{9148836}
T.~Peng, R.~Zhang, X.~Cheng, and L.~Yang, ``{LSTM}-based channel prediction for
  secure massive {MIMO} communications under imperfect {CSI},'' in \emph{ICC
  2020 - 2020 IEEE International Conference on Communications (ICC)}, 2020, pp.
  1--6.

\bibitem{9002073}
D.~Madhubabu and A.~Thakre, ``Long-short term memory based channel prediction
  for {SISO} system,'' in \emph{2019 International Conference on Communication
  and Electronics Systems (ICCES)}, 2019, pp. 1--5.

\bibitem{8395149}
A.~Alkhateeb, S.~Alex, P.~Varkey, Y.~Li, Q.~Qu, and D.~Tujkovic, ``Deep
  learning coordinated beamforming for highly-mobile millimeter wave systems,''
  \emph{IEEE Access}, vol.~6, pp. 37\,328--37\,348, 2018.

\bibitem{9048929}
M.~Alrabeiah and A.~Alkhateeb, ``Deep learning for {TDD} and {FDD} massive
  {MIMO}: Mapping channels in space and frequency,'' in \emph{2019 53rd
  Asilomar Conference on Signals, Systems, and Computers}, 2019, pp.
  1465--1470.

\bibitem{8645463}
X.~Li, A.~Alkhateeb, and C.~Tepedelenlioğlu, ``Generative adversarial
  estimation of channel covariance in vehicular millimeter wave systems,'' in
  \emph{2018 52nd Asilomar Conference on Signals, Systems, and Computers},
  2018, pp. 1572--1576.

\end{thebibliography}

\begin{IEEEbiography}[{\includegraphics[width=1in,height=1.25in,clip,keepaspectratio]{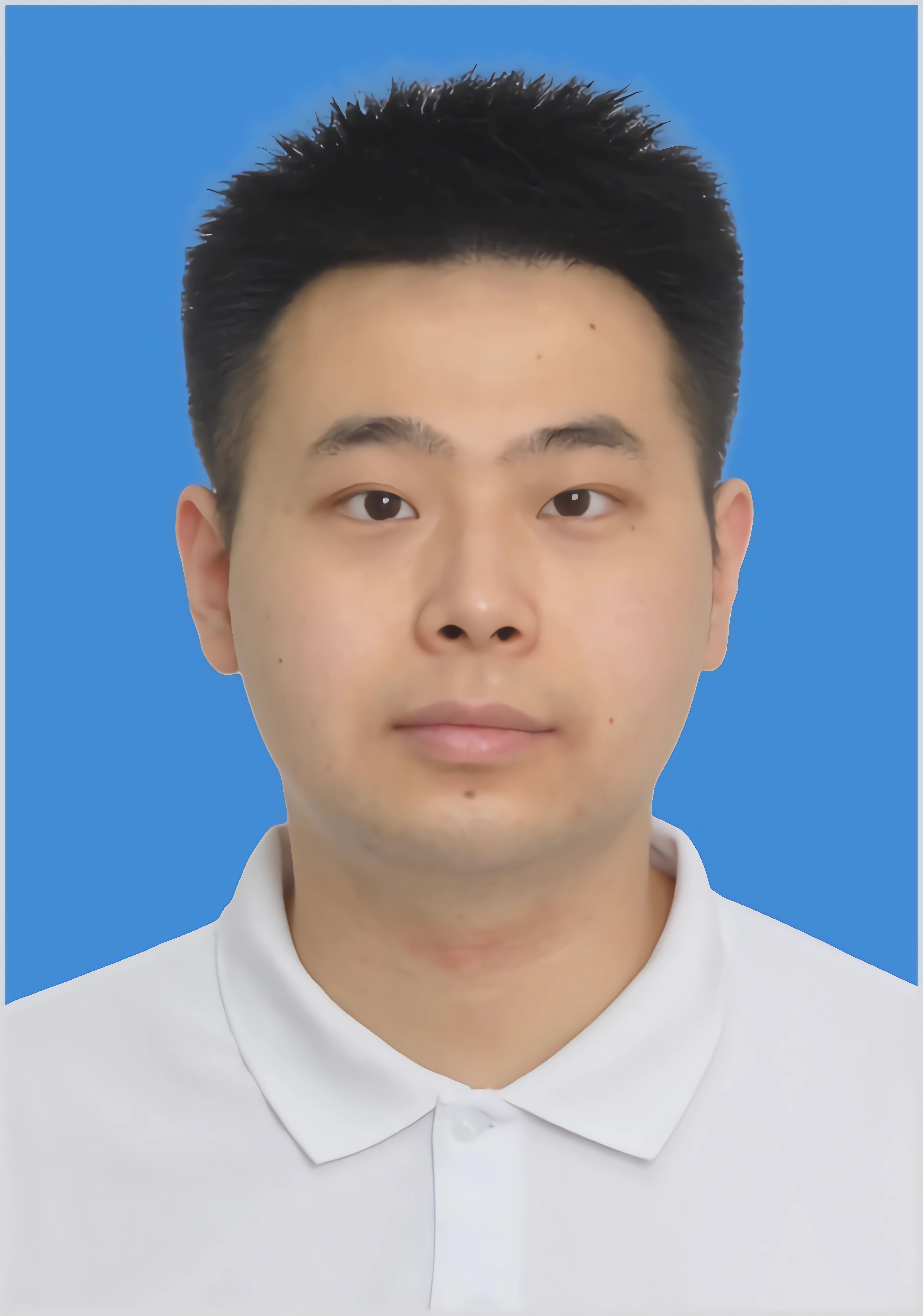}}]{Zhuoran Xiao}
	(Student Member, IEEE) received the B.S. degree in information engineering from Huazhong University of Science and Technology, Wuhan, China, in 2018, and the Ph.D. degree in information and communication engineering from Zhejiang University, Hangzhou, China in 2023, under the supervision of Prof. Zhaoyang Zhang. He is currently the Research Scientist with the Nokia Bell Labs China. His current research interests include AI/ML for next-gen communication and massive MIMO.
\end{IEEEbiography}

\begin{IEEEbiography}[{\includegraphics[width=1in,height=1.25in,clip,keepaspectratio]{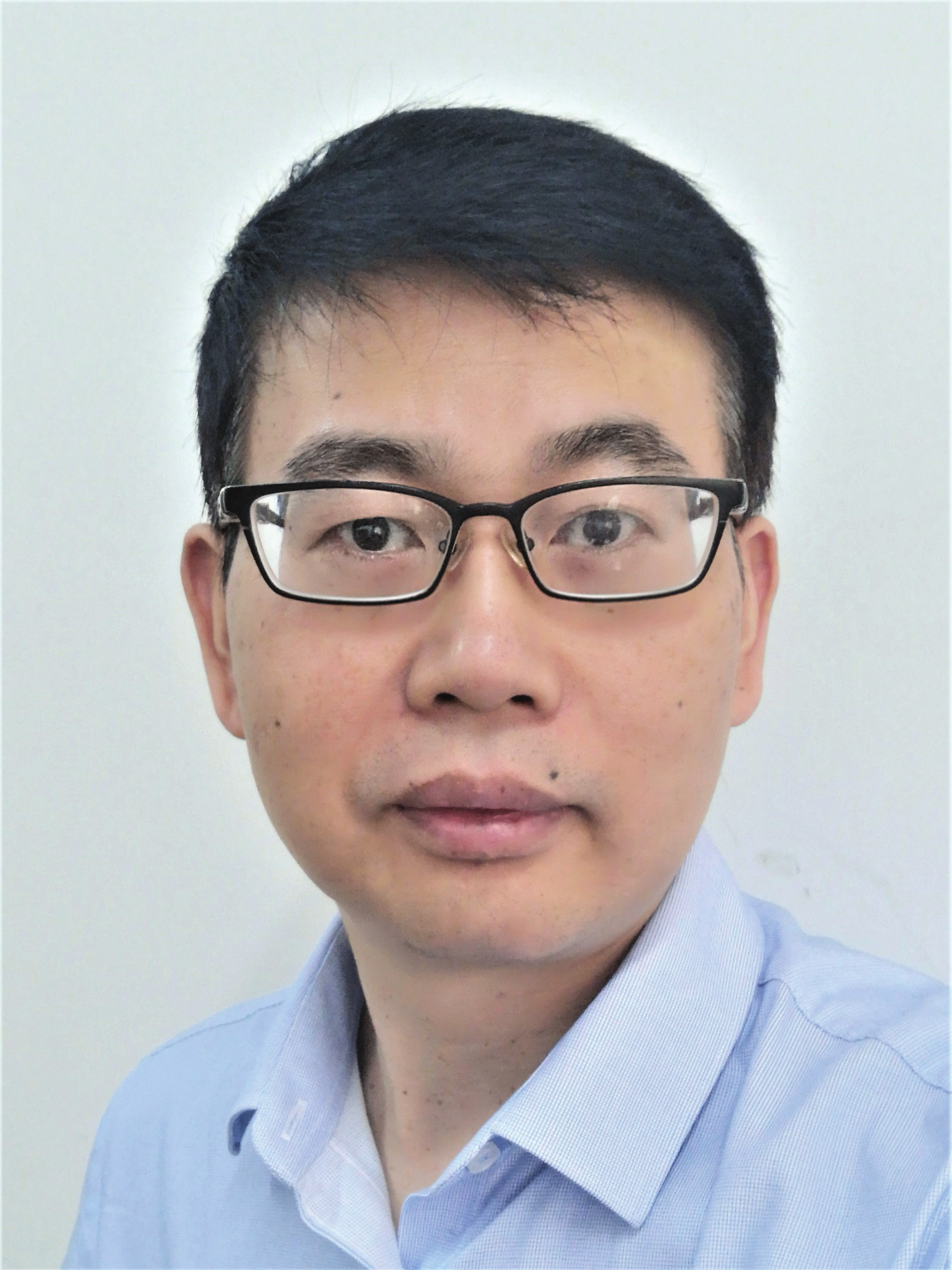}}]{Zhaoyang~Zhang} (Senior Member, IEEE) received his Ph.D. degree from Zhejiang University, Hangzhou, China, in 1998, where he is currently a Qiushi Distinguished Professor. His research interests are mainly focused on the fundamental aspects of wireless communications and networking, such as information theory and coding theory, networked signal processing and distributed learning, AI-empowered communications and networking, and integrated communication, sensing and computing, etc. He has co-authored more than 400 peer-reviewed international journal and conference papers, including 8 conference best papers from IEEE ICC 2019 and IEEE Globecom 2020, etc. He was awarded the National Natural Science Fund for Distinguished Young Scholars by NSFC in 2017, and was a co-recipient of the First Grade State-level Teaching Award for Graduate Education in 2023. 
	
Dr. Zhang is serving or has served as Editor for \textsc{IEEE Transactions on Wireless Communications} and \textsc{IEEE Transactions on Communications}, etc., as Lead Guest Editor for \textsc{IEEE Wireless Communications} Special Issue on Sustainable Big AI Model for Wireless Networks, and as General Chair, TPC Co-Chair or Symposium Co-Chair for WCSP 2023/2018/2013, PIMRC 2021 Workshop on Native AI Empowered Wireless Networks, VTC-Spring 2017 Workshop on HMWC, Globecom 2014 Wireless Communications Symposium, etc. He was also a keynote speaker for IEEE Globecom 2021 Workshop on Native-AI Wireless, APCC 2018 and VTC-Fall 2017 Workshop on NOMA, etc.
\end{IEEEbiography}

\begin{IEEEbiography}[{\includegraphics[width=1in,height=1.25in,clip,keepaspectratio]{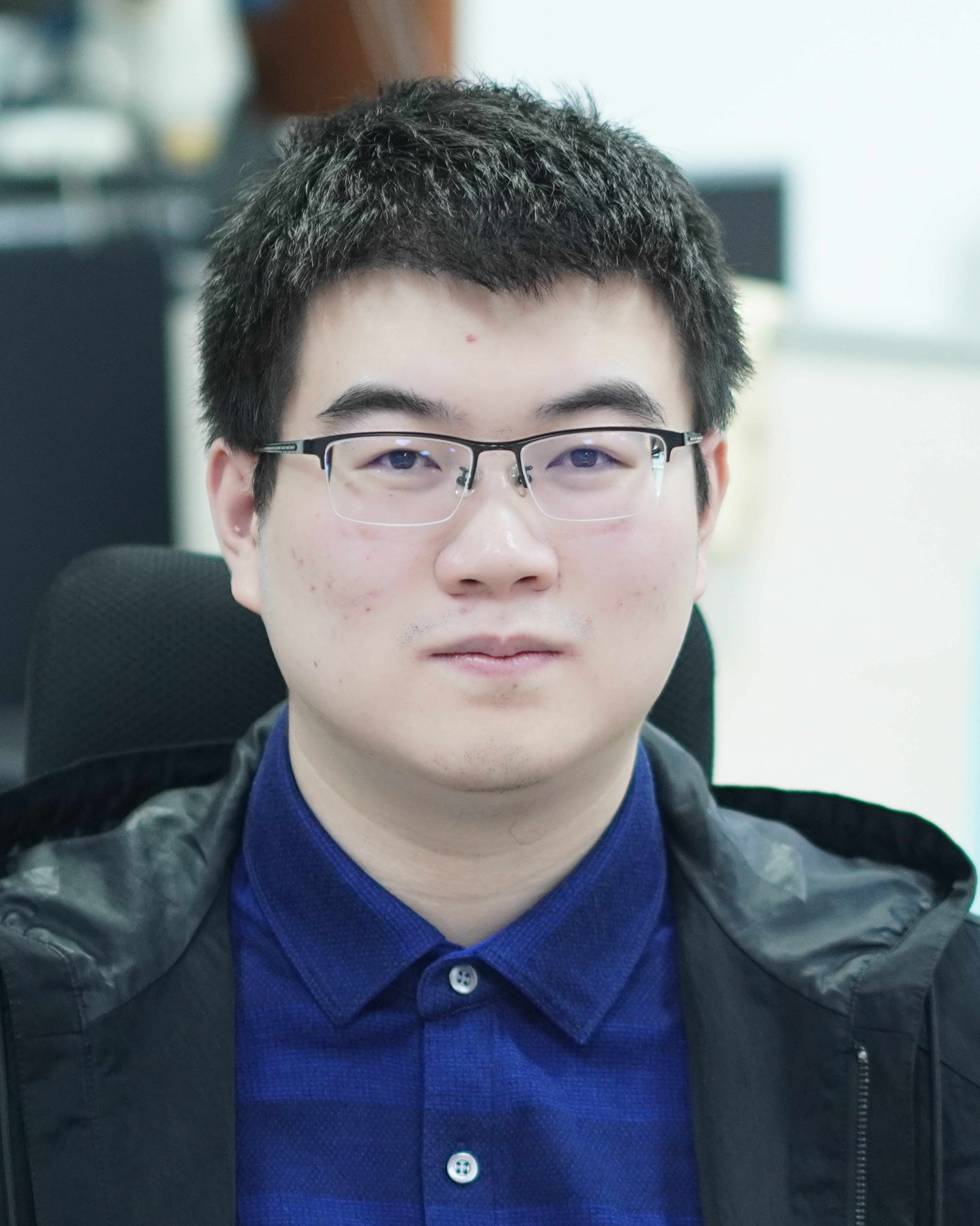}}]{Zirui Chen}
	(Student Member, IEEE) received the B.S.Eng. degree in information engineering from Zhejiang University, Hangzhou, China, in 2021. He is currently pursuing the Ph.D. degree in information and communication engineering with Zhejiang University under the supervision of Prof. Zhaoyang Zhang. His current research interests include AI-empowered communications and massive MIMO.
\end{IEEEbiography}

\begin{IEEEbiography}
	[{\includegraphics[width=1in,height=1.25in,clip,keepaspectratio]{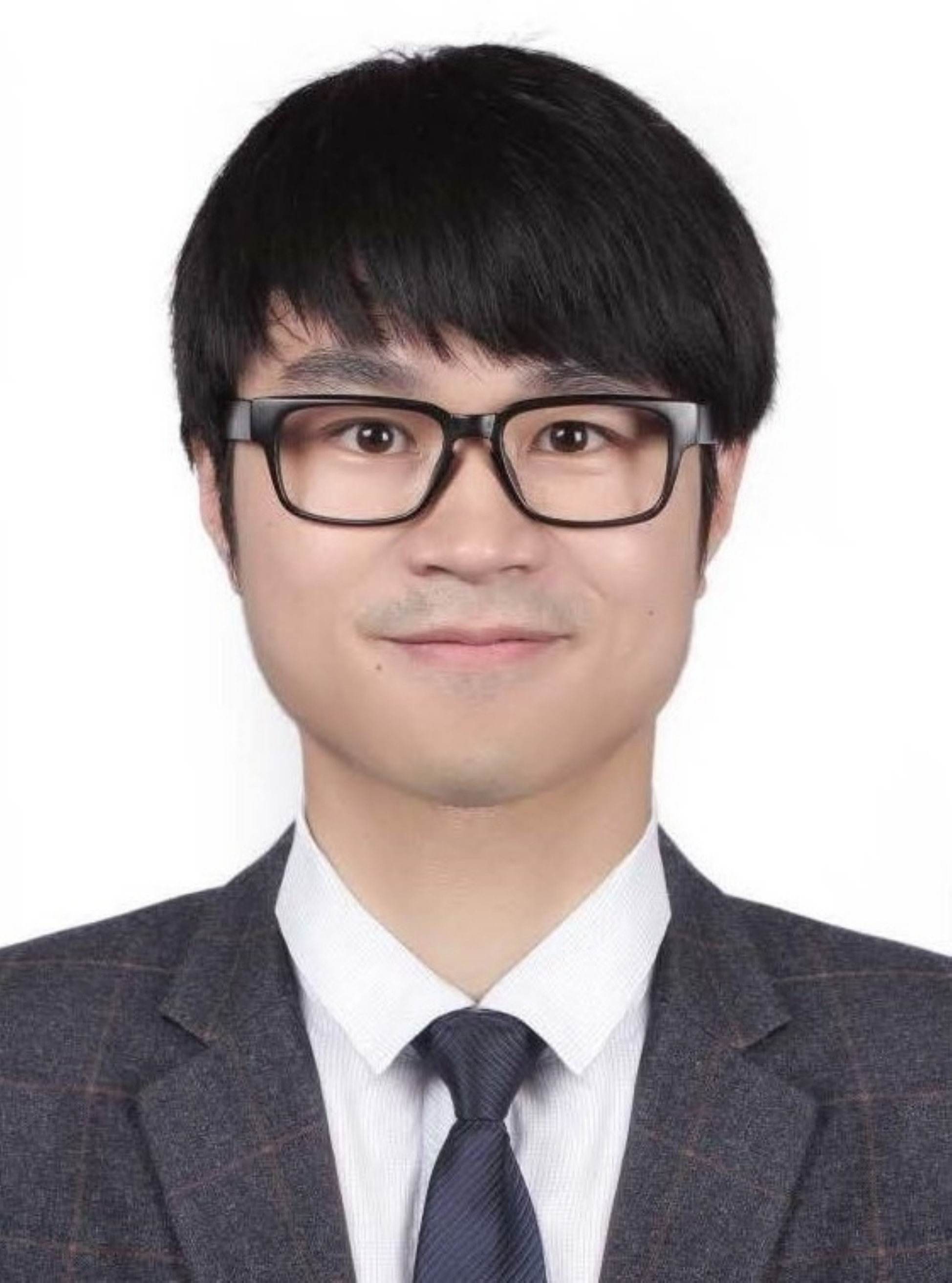}}]{Zhaohui~Yang} (Member, IEEE) received the Ph.D. degree from Southeast University, Nanjing, China, in 2018. From 2018 to 2020, he was a Post-Doctoral Research Associate at the Center for Telecommunications Research, Department of Informatics, King’s College London, U.K. From 2020 to 2022, he was a Research Fellow at the Department of Electronic and Electrical Engineering, University College London, U.K. He is currently a ZJU Young Professor with the Zhejiang Key Laboratory of Information Processing Communication and Networking, College of Information Science and Electronic Engineering, Zhejiang University, and also a Research Scientist with Zhejiang Laboratory. His research interests include joint communication, sensing, and computation, federated learning, and semantic communication. He received the 2023 IEEE Marconi Prize Paper Award, 2023 IEEE Katherine Johnson Young Author Paper Award, 2023 IEEE ICCCN best paper award, and the first prize in Invention and Entrepreneurship Award of the China Association of Inventions. He was the Co-Chair for international workshops with more than ten times including IEEE ICC, IEEE GLOBECOM, IEEE WCNC, IEEE PIMRC, and IEEE INFOCOM. He is an Associate Editor for the IEEE Communications Letters, IET Communications, and EURASIP Journal on Wireless Communications and Networking. He has served as a Guest Editor for several journals including IEEE Journal on Selected Areas in Communications.
\end{IEEEbiography}

\begin{IEEEbiography}[{\includegraphics[width=1in,height=1.25in,clip,keepaspectratio]{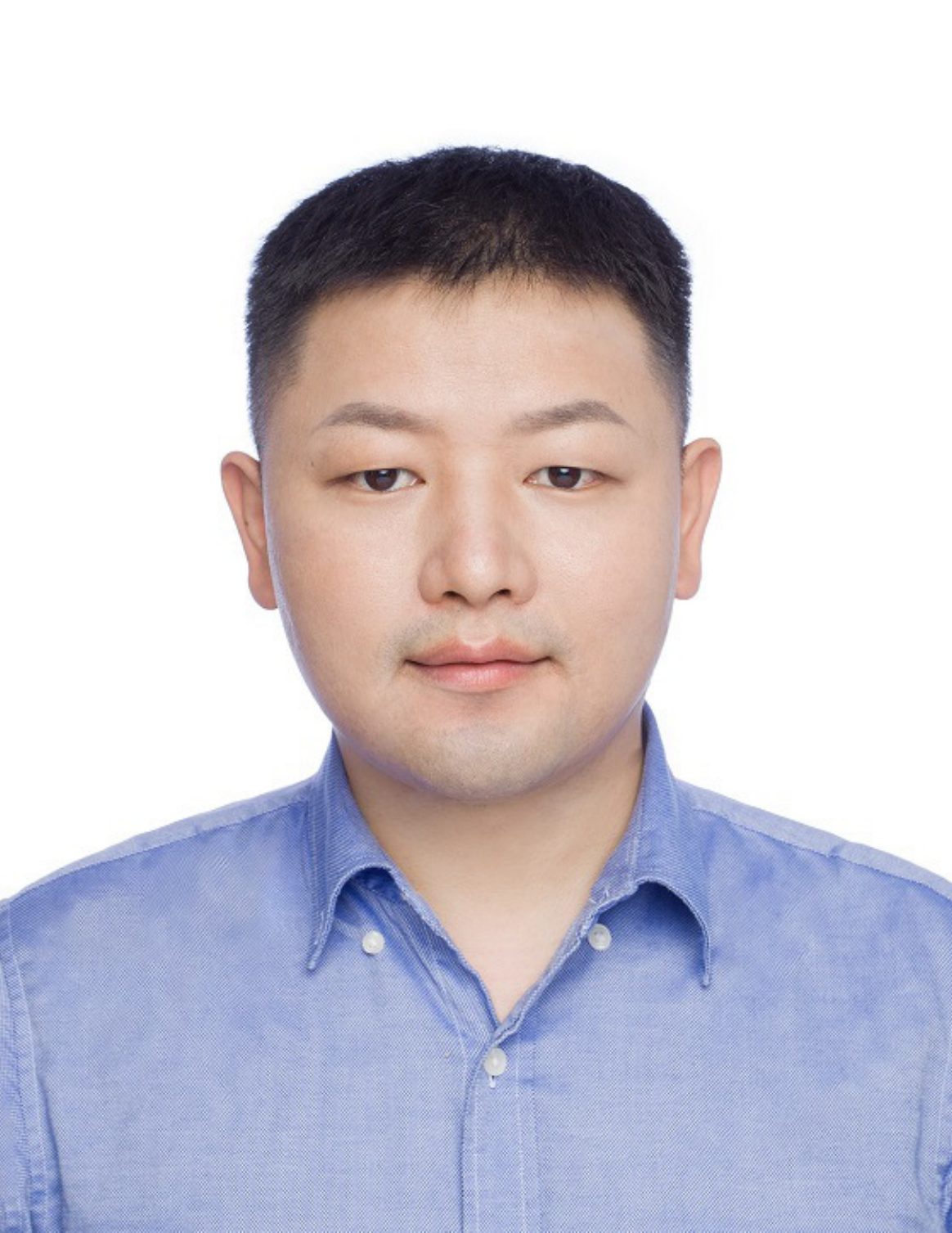}}]{Chongwen Huang} (Member, IEEE) obtained his B. Sc. degree in 2010 from Nankai University, and the M.Sc degree from the University of Electronic Science and Technology of China in 2013, and PhD degree from Singapore University of Technology and Design (SUTD) in 2019. From Oct. 2019 to Sep. 2020, he is a Postdoc in SUTD.  Since Sep. 2020, he joined into Zhejiang University as a tenure-track young professor. Dr. Huang is the recipient of 2021 IEEE Marconi Prize Paper Award, 2023 IEEE Fred W. Ellersick Prize Paper Award and 2021 IEEE ComSoc Asia-Pacific Outstanding Young Researcher Award. He has served as an Editor of IEEE Communications Letter, Elsevier Signal Processing, EURASIP Journal on Wireless Communications and Networking and Physical Communication since 2021. His main research interests are focused on Holographic MIMO Surface/Reconfigurable Intelligent Surface, B5G/6G Wireless Communications, mmWave/THz Communications, Deep Learning technologies for Wireless communications, etc.
\end{IEEEbiography}

\begin{IEEEbiography}[{\includegraphics[width=1in,height=1.25in,clip,keepaspectratio]{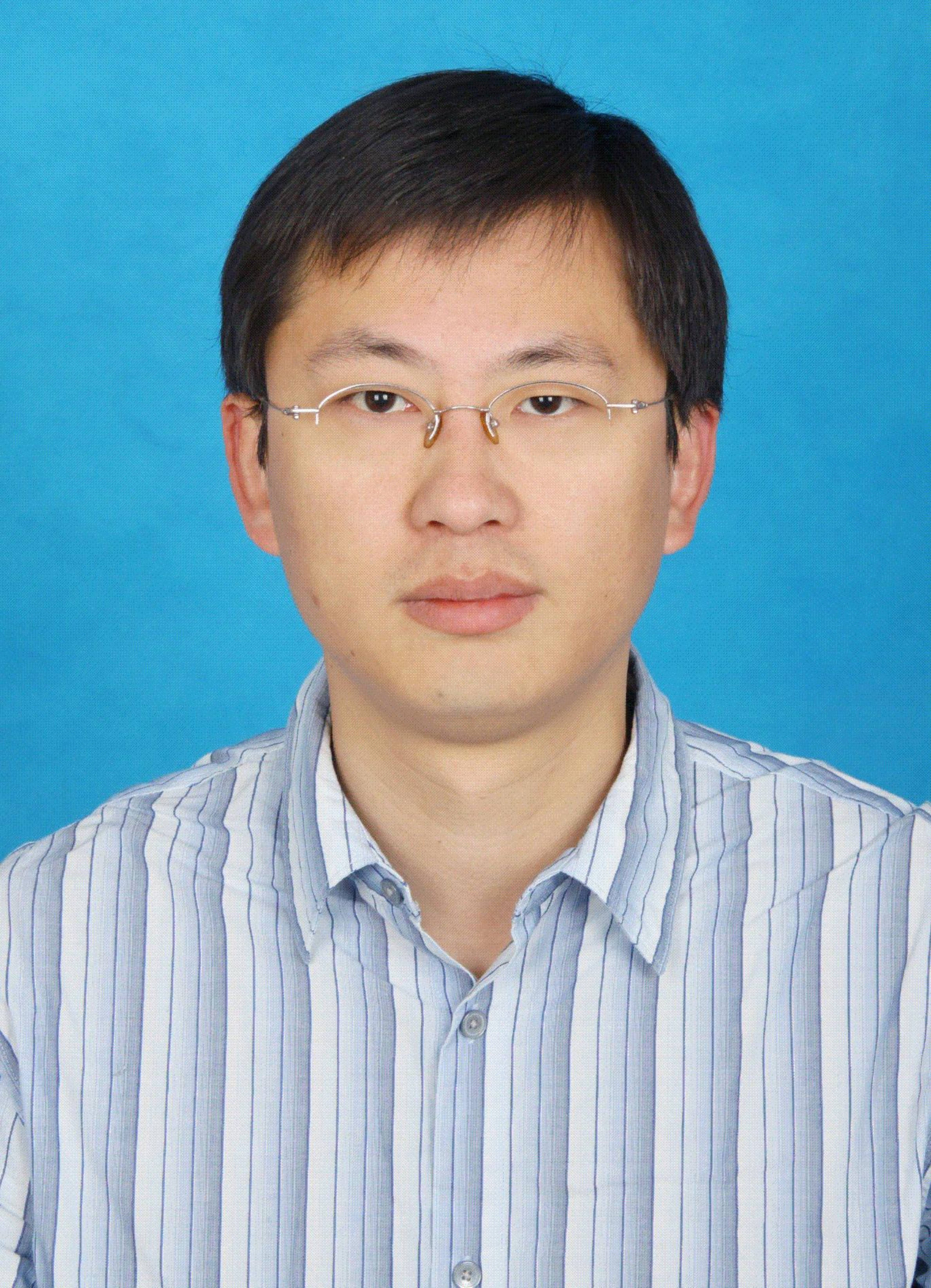}}]{Xiaoming Chen} (Senior Member, IEEE)
	received the B.Sc. degree from Hohai University in 2005, the M.Sc. degree from Nanjing University of Science and Technology in 2007 and the Ph. D. degree from Zhejiang University in 2011, all in electronic engineering. He is currently a Professor with the College of Information Science and Electronic Engineering, Zhejiang University, Hangzhou, China. From March 2011 to October 2016, He was with Nanjing University of Aeronautics and Astronautics, Nanjing, China. From February 2015 to June 2016, he was a Humboldt Research Fellow at the Institute for Digital Communications, Friedrich-Alexander-University Erlangen-N\"urnberg (FAU), Germany. His research interests mainly focus on LEO satellite constellation, Internet of Things, and smart communications.
	
	Dr. Chen served as an Editor for the \textsc{IEEE Transactions on Communications} and the \textsc{IEEE Communications Letters}, and a Guest Editor for the \textsc{IEEE Journal on Selected Areas in Communications} ``Massive Access for 5G and Beyond" and the \textsc{IEEE Wireless Communications} ``Massive Machine-Type Communications for IoT". He received the Best Paper Awards at the IEEE Global Communications Conference (GLOBECOM) 2020, the International Conference on Wireless Communications and Signal Processing (WCSP) 2020, the IEEE International Conference on Communications (ICC) 2019, and the IEEE/CIC International Conference on Communications in China (ICCC) 2018.
\end{IEEEbiography}

\end{document}